\documentclass[a4paper]{article}
\usepackage{fullpage}
\usepackage{caption}
\usepackage{subcaption}
\usepackage{graphicx} 
\usepackage{physics}
\usepackage{mathtools}
\usepackage{amsfonts, amsmath, amssymb, amsthm}
\usepackage[linesnumbered,ruled,vlined]{algorithm2e}

\SetCommentSty{mycommfont}
\usepackage{tikz,ifthen}
\usepackage{multirow}
\usepackage{array}
\usepackage{booktabs}
\usepackage{comment}
\usepackage[sort]{cite}
\usepackage{version}
\usepackage{hyperref}
\usepackage{cleveref}
\usepackage{enumerate}
\usepackage{thm-restate}

\theoremstyle{plain}
\newtheorem{theorem}{Theorem}[section]

\theoremstyle{definition}

\newtheorem{definition}[theorem]{Definition}

\theoremstyle{remark}

\theoremstyle{plain}
\newtheorem*{theorem*}{Theorem}
\newtheorem*{lemma*}{Lemma}
\newtheorem*{prop*}{Proposition}
\newtheorem*{cor*}{Corollary}
\newtheorem*{conj*}{Conjecture}

\theoremstyle{definition}
\newtheorem*{ass*}{Assumption}
\newtheorem*{definition*}{Definition}

\usepackage{stmaryrd}
\newcommand{\cint}[1]{\llbracket#1\rrbracket_k}
\newcommand{\bigcint}[1]{\big\llbracket#1\big\rrbracket_k}

\usetikzlibrary{shapes.geometric}
\usetikzlibrary{shapes.misc}
\usetikzlibrary{positioning}
\usetikzlibrary{calc}

\newcommand{\argmax}{\mathop{\rm arg\,max}} 
\newcommand{\argmin}{\mathop{\rm arg\,min}}

\newcommand{\ALG}{\mathrm{ALG}}
\newcommand{\OPT}{\mathrm{OPT}}

\usepackage{silence}
\title{Online Matching with Delays and Size-based Costs}

\usepackage{orcidlink}
\author{
    Yasushi Kawase\orcidlink{0000-0001-5626-779X}\thanks{The University of Tokyo, \texttt{kawase@mist.i.u-tokyo.ac.jp}}
\and
    Tomohiro Nakayoshi\thanks{The University of Tokyo, \texttt{nakayoshi-tomohiro@g.ecc.u-tokyo.ac.jp}}
}

\date{}

\begin{document}

\includeversion{full}\excludeversion{short}

\maketitle

\begin{abstract}
In this paper, we introduce the problem of Online Matching with Delays and Size-based Costs (OMDSC).
The OMDSC problem involves $m$ requests arriving online. At any time, a group can be formed by matching any number of these requests that have been received but are still unmatched.
The cost associated with each group is determined by the waiting time for each request within the group and a size-dependent cost. Our goal is to partition all incoming requests into multiple groups while minimizing the total associated cost. 
The problem extends the TCP acknowledgment problem proposed by Dooly et al.\ (JACM 2001). It generalizes the cost model for sending acknowledgments.
%
%
%
This paper reveals the competitive ratios for a fundamental case where the range of the penalty function is limited to $0$ and $1$. We classify such penalty functions into three distinct cases: (i) a fixed penalty of $1$ regardless of group size, (ii) a penalty of $0$ if and only if the group size is a multiple of a specific integer $k$, and (iii) other situations. The problem of case (i) is equivalent to the TCP acknowledgment problem, for which Dooly et al.\ proposed a $2$-competitive algorithm.
For case (ii), we first show that natural algorithms that match all the remaining requests are $\Omega(\sqrt{k})$-competitive. We then propose an $O(\log k / \log \log k)$-competitive deterministic algorithm by carefully managing match size and timing, and we also prove its optimality. For case (iii), we demonstrate the non-existence of a competitive online algorithm. Additionally, we discuss competitive ratios for other typical penalty functions.
\end{abstract}

\section{Introduction}
\label{sec:introduction}

In the realm of online gaming platforms, the challenge of matching players for an optimal gaming experience lies in balancing minimal waiting times and high match quality. 
Consider, for instance, an online gaming platform hosting a $k$-players game, such as chess ($k=2$), Mahjong ($k=4$), and battle royal-style shooting games (e.g. $k=60$ for Apex Legends).
On such platforms, players sequentially enter a waiting queue, and the platform selects $k$ players from the queue to start a match.
To address scenarios with an insufficient number of players, the platform may opt to fill a group with computer (AI) players. While this approach ensures prompt matchmaking, it potentially compromises the quality of the gaming experience. Conversely, waiting to gather the full quota of required $k$ players may significantly increase waiting time, potentially reducing player satisfaction.



Motivated by the above challenge, we introduce the \emph{Online Matching with Delays and Size-based Costs (OMDSC)} problem. 
In this problem, requests (which correspond to players) are presented sequentially in real time.
At any time, it is possible to match a group consisting of any number of previously received but unmatched requests.
A key aspect of the problem is the \emph{penalty function}, which determines a cost based on the size of the group.
The total cost associated with each match is defined by the waiting time for each request within the group and a cost that depends on the size of the group.



Applications of the OMDSC problem extend beyond gaming, including online shopping and web services that allow batch-processing of requests.
For instance, in online shopping services, a higher penalty may be imposed for dispatching delivery trucks from warehouses with fewer orders to reduce the required resources, such as trucks and drivers.
Similarly, web services offering AI features can process multiple inputs in parallel using GPU batch processing. Such services benefit from processing a larger number of requests at once compared to processing fewer requests at once. 
Therefore, our study is also beneficial for applications where the cost difference due to the quantity of requests has a minor impact on the execution cost.


It should be noted that our problem is closely related to the TCP acknowledgment problem~\cite{Dooly2001-yl}, the online weighted cardinality \emph{Joint Replenishment Problem (JRP)} with delays~\cite{Chen2022-qj}, and the online \emph{Min-cost Perfect Matching with Delays (MPMD)}~\cite{Emek2016-sa}.
The TCP acknowledgment problem (without look-ahead) 
is equivalent to the OMDSC problem with a penalty function being constant.
There are various generalizations of the TCP acknowledgment problem that partially capture the OMDSC problem. Specifically, the online weighted cardinality JRP with delays considers costs dependent on the (weighted) cardinality of each group. In those generalizations, it is assumed that the penalty function is monotonically non-decreasing. However, the OMDSC problem treats penalty functions that are not monotone.
In the MPMD problem, requests arrive on a metric space, and matching is restricted to groups of exactly size $2$. The online \emph{Min-cost Perfect $k$-way Matching with Delays ($k$-MPMD)}~\cite{Melnyk2021-zv} extends the group size to $k$.
For more details, see \cref{subsec:related}



\subsection{Our Results}

In this paper, we determine the competitive ratios of the OMDSC problem for a fundamental and critical scenario in which the penalty function takes only values of either $0$ or $1$.
In the OMDSC problem, dividing a group into multiple smaller groups could reduce its penalty. Therefore, it is sufficient to consider a modified penalty function that can be obtained by optimally dividing a group into subgroups.
With this modification, we classify penalty functions into three cases: 
(i) the (modified) penalty is $1$, regardless of the group size,
(ii) the modified penalty is $0$ if and only if the group size is a multiple of a specific integer $k$, and
(iii) functions not covered by the above cases.


For case (i), the OMDSC problem is equivalent to the TCP acknowledgment problem. Dooly et al.~\cite{Dooly2001-yl} proposed a $2$-competitive online algorithm that matches all remaining requests when the waiting cost increases by $1$. 
For case (ii), we first examine natural algorithms that 
match all remaining requests when performing a costly match,
which we refer to as \emph{match-all-remaining} algorithms.
We demonstrate that every match-all-remaining algorithm is $\Omega(\sqrt{k})$-competitive (\Cref{thm:match_all_remaining}).
Consequently, to obtain an $o(\sqrt{k})$-competitive algorithm, we need to consider both the timing and the size of requests to be matched.
We propose an $O(\log k/\log\log k)$-competitive online algorithm by carefully managing the number of remaining requests (\Cref{thm:upper}). 
We also prove that this competitive ratio is best possible (\Cref{thm:lower_bound}).
Note that the competitive ratio for case (ii) with $k=1$ is trivial (\Cref{thm:k_is_1}) and with $k = 2$ can be obtained from a result by Emek et al.~\cite{Emek2019-if}.
For case (iii), we prove that no competitive online algorithm exists (\Cref{thm:no_multiple}).
The results are summarized in \Cref{tbl:result}.
Furthermore, we discuss competitive ratios for other typical penalty functions in \Cref{subsec:other_penalty}.


\begin{table}[htbp]
    \caption{The competitive ratios for the OMDSC problem}\label{tbl:result}
    \centering
    \begin{tabular}{cc|cc}
        \toprule
        \multicolumn{2}{c|}{\textbf{Penalty function} (with modification)} & \textbf{Competitive ratio} & \textbf{Reference}\\
        \midrule
        (i) always $1$  & & $2$  & Dooly et al.~\cite{Dooly2001-yl}\\\hline
             & $k=1$       & 1  & \Cref{thm:k_is_1} \\\cline{2-4}
        (ii) $0$ if the size is a multiple of $k$ & $k=2$       & $3$ & Emek et al.~\cite{Emek2019-if}  \\\cline{2-4}
             & general $k$ & $\Theta\left(\log k / \log \log k\right)$ & \Cref{thm:upper,thm:lower_bound}\\\hline
        (iii) other scenarios & & unbounded & \Cref{thm:no_multiple} \\
        \bottomrule
    \end{tabular}
\end{table}


\subsection{Related Work}\label{subsec:related}

An online version of matching problems was first introduced by Karp et al.~\cite{Karp1990-kk}. In their study, arriving requests need to be matched immediately upon arrival, primarily focusing on bipartite matching.
Additionally, research has been conducted on minimizing matching costs, both in settings where vertices have determined costs and in settings that consider distances in a metric space.
Mehta et al.~\cite{Mehta2007-cl} proposed applications to AdWords.
Other applications, such as food delivery, have also been considered. 
For more details, see \cite{Mehta2013-lp,Echenique2023-pd}.


On the other hand, situations like online game matchmaking and ride-sharing services allow for delaying user matches. Emek et al.~\cite{Emek2016-sa} modeled this scenario as the MPMD problem. In this setting, requests can be delayed by imposing waiting costs. 
Several online algorithms have been proposed for the MPMD problem~\cite{Azar2017-vj, Bienkowski2018-yt, Azar2020-in}.
Subsequently, Melnyk et al.~\cite{Melnyk2021-zv} extended the MPMD problem to the $k$-MPMD problem, which requires to make matchings of size exactly $k$. For this problem, both deterministic and randomized algorithms have been proposed~\cite{Melnyk2021-zv, Kakimura2023-im}. 

Dooly et al.~\cite{Dooly2001-yl} introduced the TCP acknowledgment problem, which concerns the acknowledgment of TCP packets. In this problem, we aim to minimize the sum of the costs for acknowledgments and the costs for delaying TCP packets. The optimal deterministic and randomized algorithms were proposed by Dooly et al.~\cite{Dooly2001-yl} and Karlin et al.~\cite{Karlin2001-fx}, respectively. The TCP acknowledgment problem is equivalent to the OMDSC problem, where the penalty function corresponds to case (i). Conversely, the OMDSC problem can be viewed as a generalization of the acknowledgment cost in the TCP acknowledgment problem.

Various extensions to the TCP acknowledgment problem have been proposed~\cite{Buchbinder2013-yc, Bienkowski2014-fz, Bienkowski2016-dk, Buchbinder2017-qx, Azar2019-sn, Carrasco2018-vp, Azar2020-yp, Touitou2021-dx, Chen2022-qj}.
Specifically, Chen et al.~\cite{Chen2022-qj} introduced the problem of online weighted cardinality JRP with delays. Their model can handle a penalty based on the size of a match. However, unlike our study, their penalty function is restricted to monotonically non-decreasing.
They proposed a constant competitive algorithm for the problem.


The objectives of the MPMD problem and the $k$-MPMD problem are to pair requests and form groups of exactly size $k$, respectively. In contrast, our proposed OMDSC problem does not impose constraints on the number of requests in each group.
Emek et al.~\cite{Emek2016-sa} also introduced a problem called MPMDfp, which allows clearing a request at a cost.
The MPMDfp problem on a single source is equivalent to the OMDSC problem where the penalty function corresponds to case (ii) with $k = 2$.
Additionally, the MPMDfp problem on a single source can be reduced to the MPMD problem on two sources~\cite{Emek2016-sa}.
For the MPMD on two sources, Emek et al.~\cite{Emek2019-if} proposed a deterministic algorithm, and He et al.~\cite{He2023-ot} proposed a randomized algorithm (see \Cref{sec:case_ii} for more details).


Another approach to extending the MPMD problem is by modifying waiting costs.
Liu et al.~\cite{Liu2018-wb} generalized waiting costs from linear to convex. 
Other variations of waiting costs have also been studied~\cite{Liu2022-bu, Azar2021-bb, Deryckere2023-oj}.
Pavone et al.~\cite{Pavone2022-cg} studied the \emph{Online Hypergraph Matching with Deadlines} problem. Though hypergraph matching does not impose constraints on group size, it differs from the OMDSC problem in that it employs deadlines instead of delays and requests arrive one at each unit of time.


\section{Preliminaries}


We denote the set of real numbers as $\mathbb{R}$, the set of integers as $\mathbb{Z}$, the set of non-negative integers as $\mathbb{Z}_+$, and the set of positive integers as $\mathbb{Z}_{++}$. Additionally, for a positive integer $k \in \mathbb{Z}_{++}$, we define $\mathbb{Z}_k$ as the set of integers from $0$ to $k - 1$.


\subsection{Problem Settings}
In this section, we formally define the OMDSC problem.
An instance of the problem is specified by a penalty function $f\colon\mathbb{Z}_{++}\to\mathbb{R}$ and $m$ requests $V = \{v_1, v_2, \ldots, v_m\}$, which arrive in real-time. 
Each request $v_i$ arrives at time $t_i$, where $0 \leq t_1 \leq t_2 \leq \cdots \leq t_m$ is assumed.

An online algorithm initially knows only the function $f$; it does not have prior knowledge of the arrival times and the total number of requests. Each request $v_i$ is revealed at time $t_i$.
At each time $\tau$, the algorithm can match any subset $S_j$ of requests that appeared by time $\tau$ and have not yet been matched. 
The cost to match requests in $S_j$ at time $\tau_j$ is defined by 
\begin{align*}
\textstyle
    f(|S_j|) + \sum_{v_i \in S_j} \left( \tau_j - t_i \right),
\end{align*}
where the first term represents the \emph{size cost} and the second term represents the \emph{waiting cost}.
In addition, the \emph{waiting cost} at time $\tau$ is defined as $\sum_{v_i\in V'}(\tau-t_i)$, where $V'$ is the set of unmatched but already presented requests at that time.


The objective is to design an online algorithm that matches all requests while minimizing the total cost.
We use the \emph{competitive ratio} to evaluate the performance of an online algorithm.
For a given instance $\sigma$, let $\ALG(\sigma)$ represent the cost incurred by the online algorithm, and let $\OPT(\sigma)$ denote the optimal cost with prior knowledge of all requests in the instance.
The competitive ratio of $\ALG$ for an instance $\sigma$ is defined as $\ALG(\sigma)/\OPT(\sigma)$, interpreting $0 / 0$ as $1$.
Additionally, the competitive ratio of $\ALG$ for a problem is then defined as the supremum of the competitive ratio over all instances, i.e., $\sup_{\sigma} \ALG(\sigma)/\OPT(\sigma)$.

Note that the competitive ratio defined above is the \emph{strict} competitive ratio. We can also define the \emph{asymptotic} competitive ratio as $\limsup_{\ALG(\sigma)\to\infty}\ALG(\sigma)/\OPT(\sigma)$. However, in our problem, the strict and asymptotic competitive ratios of optimal algorithms coincide. Indeed, if an algorithm is strictly $\rho$-competitive, then it is also asymptotically at most $\rho$-competitive by definition. Moreover, if no strictly $\rho$-competitive algorithm exists, we can construct an instance for each algorithm where the strict competitive ratio is at most $\rho$. Thus, by repeatedly constructing and providing such instances, we can conclude that no algorithm has an asymptotic competitive ratio better than $\rho$.
Hence, we only deal with the strict competitive ratio hereafter.






\subsection{Binary Penalty Function}

This study primarily focuses on penalty functions where the range is $\{0, 1\}$.
Other penalty functions will be discussed in \Cref{subsec:other_penalty}.
For such a binary penalty function $f$, we may be able to match $n$ requests with the size cost $0$, even if $f(n)=1$, by appropriately partitioning the requests. 
For instance, if $f(2)=f(3)=0$ and $f(7)=1$, we can match $7$ requests with a size cost of $0$ by partitioning them into groups of sizes $2$, $2$, and $3$.
We introduce a notation for the set of numbers of requests that can be matched with a size cost of $0$ as follows.
\begin{definition}
    For a penalty function $f \colon \mathbb{Z}_{++} \to \{0, 1\}$, 
    we define the zero-penalty set $B_f$ as the set of quantities that can be matched with a size cost of $0$ by optimally dividing the requests into some subsets and matching them. Formally,
    \[
    B_f \coloneqq \big\{n \in \mathbb{Z}_{++} \;\big|\; \exists n_1, n_2, \ldots, n_s\in\mathbb{Z}_{++} \text{ \rm{s.t.} } \textstyle\sum_{i=1}^s n_i = n,~f(n_i) = 0~(\forall i \in \{1, 2, \ldots, s\})\big\}.
    \]
\end{definition}
For example, if $f(1)=0$, then $B_f=\mathbb{Z}_{++}$. 
Alternatively, if $f(1)=1$ and $f(2)=f(3)=0$, then $B_f=\mathbb{Z}_{++}\setminus\{1\}$.
We can interpret that the size cost for matching $n$ requests is $0$ if $n\in B_f$ and $1$ if $n\not\in B_f$.



We classify the binary penalty functions into three types as follows:
(i) $B_f = \emptyset$,
(ii) $B_f = \{kn \mid n \in \mathbb{Z}_{++}\}$ for some $k\in\mathbb{Z}_{++}$, and
(iii) otherwise.
The case of (i) is a situation where the penalty for matching requests is always $1$ (i.e., $f(x)=1$ for all $x\in\mathbb{Z}_{++}$).
The case of (ii) is a situation where a set of requests can be matched without size cost only if the size is a multiple of $k=\min\{x\in\mathbb{Z}_{++}\mid f(x)=0\}$. This case has an application in the context of an online gaming platform that hosts a $k$-players game, as described in the Introduction.



\subsection{Other Penalty Functions}\label{subsec:other_penalty}

If the range of the penalty function is $\{0,\mu\}$ with a positive real $\mu$, we can treat it as a binary penalty function by scaling increase rates of waiting costs to be $\mu$ times slower. This is because the cost to match requests in $S_j$ at time $\tau_j$ can be written as $\mu\cdot \big(f(|S_j|)/\mu+\sum_{v_i\in S_j}(\tau_j/\mu-t_i/\mu)\big)$.
For example, if $\mu = 60$ and the unit time is one minute in the original problem, then the range of the penalty function can be treated as $\{0,1\}$ by adjusting the unit time to one hour.
Let $\mu,\lambda\in\mathbb{R}$ with $0<\mu\le \lambda$, and consider a penalty function that takes values of $0$ or within the range $[\mu,\lambda]$. By applying our $\rho$-competitive algorithm to the problem, considering only whether the penalty is zero or $\mu$, we obtain a $(\rho\cdot \lambda/\mu)$-competitive algorithm. In addition, our impossibility can be transferred in the same way.

Another interesting penalty function would be $f(n) = \lceil n / k \rceil$ for a specific integer $k$.
This penalty function appears when a service can process up to $k$ requests at once, and each processing costs a fixed amount.
For this penalty function, an algorithm similar to the algorithm for (i) is $2$-competitive for any $k$\begin{short} (see \Cref{sec:appendix_2_comp})\end{short}. 

\begin{full}
\begin{short}
In this section, we examine the OMDSC problem when the penalty function satisfies $f(n)=\lceil n/k\rceil$.
\end{short}
Note that, when $k$ goes to infinity, the penalty function is equivalent to the constant function that returns $1$. Thus, the $2$-competitive algorithm for the TCP acknowledgment problem~\cite{Dooly2001-yl} works for this extreme setting. The algorithm matches all remaining requests when the waiting cost increases by $1$ after the last matching. 

We show that a similar algorithm is $2$-competitive for the penalty function with any $k$.

\begin{theorem} \label{thm:upper_for_every_k_penalty}
    For the OMDSC problem where the penalty function $f$ holds $f(n) = \lceil n / k \rceil$, with a positive integer $k\in\mathbb{Z}_{++}$, there exists a deterministic online algorithm with a competitive ratio of $2$.
\end{theorem}
\begin{proof}
    Consider the algorithm $\ALG$ that matches all remaining requests when the waiting cost increases by $1$ after the last matching and matches $k$ requests when the number of remaining requests becomes more than or equal to $k$ requests.

    For simplicity, we only consider instances $\sigma$ for which the arrival times of the requests are distinct. This restriction does not lose generality because perturbing the arrival times of simultaneously arrived requests to be slightly different has no significant effect on the costs.


    For an instance $\sigma$, let $a$ be the number of matches by $\ALG$ and $\tau_j$ be time of $j$th match where $0 = \tau_0 < \tau_1 < \tau_2 < \cdots < \tau_{a}$. Note that all match times are distinct due to the imposed restriction. In addition, all the (at most $k$) remaining requests are matched at each match.
    Then, the cost of $\ALG$ for $\sigma$ is at most $2a$ since the size cost and the waiting cost are at most one for each match.
    It is sufficient to prove that the cost of the optimal algorithm $\OPT(\sigma)$ is at least $a$. 

     
    We analyze the period $(\tau_{j-1}, \tau_{j}]$ for each $j \in {1,2,\dots,a}$ and divide it into two cases: (1) $\ALG$ matches $k$ requests at $\tau_{j}$ and (2) $\ALG$ matches fewer than $k$ requests at $\tau_{j}$. In case (1), $k$ requests arrive during the period. Since $f(k + x) = f(x) + 1$ holds for any positive integer $x \in \mathbb{Z}{++}$, keeping more than or equal to $k$ requests without matching is not an optimal behavior. Thus, $\OPT$ must match at least once in the period, incurring a size cost of at least $1$. In case (2), a waiting cost of $1$ is incurred by $\ALG$ for requests that arrive during the period. If $\OPT$ does not match during the period, it incurs a waiting cost of at least $1$. If it matches during the period, it also incurs a size cost of $1$.

    Therefore, $\OPT(\sigma)$ is at least $a$, and the competitive ratio of $\ALG$ is $2$.    
\end{proof}

In particular, the best competitive ratio is $1$ for the OMDSC problem where $f(n) = n$, that is a case of $k = 1$. An online algorithm that immediately matches each request individually upon its arrival achieves a competitive ratio of $1$.

\begin{theorem}
   For the OMDSC problem with a penalty function $f(n) = n$, there exists a deterministic online algorithm with a competitive ratio of $1$. 
\end{theorem}

\end{full}



\section{Zero-penalty Set is Nonempty and Non-representable as Multiples}
In this section, we examine case (iii), where $B_f \neq \emptyset$ and $B_f \neq \{kn \mid n \in \mathbb{Z}_{++}\}$ for any $k \in \mathbb{Z}_{++}$. 
We demonstrate that, in this case, no algorithm has a finite competitive ratio.

To illustrate this intuitively, let us consider the case where $f(1)=1$ and $f(2)=f(3)=0$, representing a poker table for at least two players.
Imagine two players arriving initially, followed potentially by a third. If we match the first two players immediately upon their arrival, matching the subsequent player incurs a cost. Alternatively, if we wait to match the first two, an unnecessary waiting cost is incurred when no additional player appears.
We can construct similar instances for every penalty function in case (iii).
\begin{theorem} \label{thm:no_multiple}
    When the penalty function $f$ satisfies both $B_f \neq \emptyset$ and $B_f \neq \{kn \mid n \in \mathbb{Z}_{++}\}$ for every $k \in \mathbb{Z}_{++}$, the competitive ratio of any randomized algorithm for the OMDSC problem is unbounded against an oblivious adversary.
\end{theorem}

\begin{proof}
    Let $k^* \coloneqq \min B_f$, where such a value must exist by the assumption that $B_f\ne\emptyset$. 
    Additionally, let $\ell\coloneqq \min B_f\setminus\{k^*n\mid n\in\mathbb{Z}_{++}\}$, where such a value must exists because $\{k^*n \mid n \in \mathbb{Z}_{++}\}\subseteq B_f$ and $B_f \neq \{k^*n \mid n \in \mathbb{Z}_{++}\}$.
    We fix an arbitrary online algorithm $\ALG$ and take a sufficiently small $\varepsilon > 0$.
    Consider an instance where $k^*$ requests are given at time $0$, and thereafter, there may be arrivals of $\ell-k^*$ additional requests depending on the behavior of $\ALG$.
    Suppose that $\ALG$ matches all the initial requests before $\varepsilon$ with probability $p$.


    If $p \geq 1/2$, 
    consider an instance where $\ell-k^*$ additional requests are given at time $\varepsilon$. Since $\ell-k^* \notin B_f$, $\ALG$ incurs an expected size cost of at least $p~(\ge 1/2)$. In contrast, the minimum total cost is $k^*\varepsilon$ by matching all $k^*+(\ell-k^*)=\ell$ requests at time $\varepsilon$. Thus, the competitive ratio is at least $\frac{p}{k^*\varepsilon}\ge\frac{1}{2k^*\varepsilon}$. 

    Conversely, if $p < 1/2$, consider the instance where no additional requests are presented. Then, the (expected) waiting cost for $\ALG$ is at least $(1 - p)k^*\varepsilon > k^*\varepsilon / 2$, while the offline optimal cost is $0$ by matching all requests at time $0$. Hence, the competitive ratio is unbounded.

    Therefore, in both scenarios, the competitive ratio is unbounded as $\varepsilon$ goes to $0$, proving that the competitive ratio for any online algorithm is unbounded.
\end{proof}


    




\section{Zero-penalty Set is Multiples of an Integer}
\label{sec:case_ii}

In this section, we investigate case (ii), where 
$B_f = \{kn \mid n \in \mathbb{Z}_{++}\}$ for $k \in \mathbb{Z}_{++}$.

When $k=1$ (i.e., $B_f=\mathbb{Z}_{++}$), we have $f(1)=0$.
Thus, the algorithm that immediately matches each request individually upon its arrival incurs a cost of $0$ and is $1$-competitive.
\begin{theorem} \label{thm:k_is_1}
    For the OMDSC problem with a penalty function $f$ such that $B_f = \mathbb{Z}_{++}$, there exists a deterministic online algorithm with a competitive ratio of $1$.
\end{theorem}

For the case $k = 2$, the OMDSC problem corresponds to the problem of MPMDfp in a metric space of a single point, which can be reduced to the problem of MPMD for two sources~\cite{Emek2016-sa} by setting the distance between two sources to $2$ and giving requests for two sources simultaneously when a request is given in the OMDSC problem. Then, matching across two sources corresponds to matching only one request in the OMDSC problem. The optimal cost associated with the reduced problem is half that of the original problem. 
Emek et al.~\cite{Emek2019-if} provided a $3$-competitive online algorithm for the online MPMD problem on two sources, and demonstrated that this is best possible.
In the instances constructed for the proof of the lower bound, requests are always given to two sources simultaneously, which can be reduced to the OMDSC problem with $k = 2$. 
Therefore, we obtain the following theorem. 
\begin{theorem}[Emek et al.~\cite{Emek2019-if}]\label{thm:k_is_2}
    For the OMDSC problem with a penalty function $f$ such that $B_f = \{2n\mid n\in\mathbb{Z}_{++}\}$, there exists a $3$-competitive deterministic algorithm.
    Moreover, no deterministic online algorithm has a competitive ratio smaller than $3$.
\end{theorem}
He et al.~\cite{He2023-ot} proposed a $2$-competitive randomized algorithm against an oblivious adversary for the MPMD problem on two sources. 
Thus, this leads to a $2$-competitive randomized algorithm for the OMDSC problem with $B_f = \{2n\mid n\in\mathbb{Z}_{++}\}$.

In the remainder of this section, we conduct an asymptotic analysis with respect to $k$.

Firstly, we define a type of algorithm as below and demonstrate the difficulty of case (ii) compared to case (i).
\begin{definition} \label{def:match_all_remaining}
We call an algorithm for the OMDSC problem \emph{match-all-remaining} if, whenever it makes a match, it matches all remaining requests or the size is in $B_f$.
\end{definition}
For case (i), where a penalty to match requests is always $1$, a match-all-remaining algorithm achieves the optimal competitive ratio. 
However, for case (ii), we show that every match-all-remaining algorithm has a competitive ratio of $\Omega(\sqrt{k})$\begin{short}, where the proof can be found in \Cref{sec:simple_lower_bound}\end{short}.
\begin{restatable}{theorem}{thmmatchallremaining} \label{thm:match_all_remaining}
    For the OMDSC problem with a penalty function $f$ that satisfies $B_f = \{kn \mid n \in \mathbb{Z}_{++}\}$, every match-all-remaining algorithm has a competitive ratio of $\Omega(\sqrt{k})$.
\end{restatable}

\begin{full}
\begin{short}
In this section, we prove that the competitive ratio of every match-all-remaining algorithm is at least $\Omega(\sqrt{k})$, if the zero-penalty set is the multiples of $k$.
Recall that a match-all-remaining algorithm is an algorithm that always matches $k$ requests or all remaining requests every time.
\thmmatchallremaining*
\end{short}

\begin{proof}
    We fix a match-all-remaining algorithm $\ALG$, and construct an instance $\sigma_{\ALG}$ depending on the behavior of $\ALG$.
    The instance $\sigma_{\ALG}$ provides $k-1$ requests at time $0$.
    Additionally, it provides $k-1$ requests that arrive immediately after each match by $\ALG$ before time $1$.

    Let $a$ be the number of matches by $\ALG$. Then, the cost of $\ALG$ for $\sigma_{\ALG}$ is at least $a + (k - 1)$, where the first and second terms correspond to the size cost and the waiting cost, respectively.
    Consider an algorithm $\ALG^*$ that runs as follows:
    \begin{itemize}
        \item If there are at least $k$ remaining requests, then match $k$ requests;
        \item If there are exactly $k-1$ remaining requests, then match $k - 1 - \lfloor \sqrt{k} \rfloor$ requests;
        \item At time $1$, match all the remaining matches.
    \end{itemize}
    This algorithm always holds at most $\sqrt{k}$ requests, hence, the waiting cost is at most $\sqrt{k}$. The size costs are incurred at most $1 + \lceil a/(1+\lfloor \sqrt{k}\rfloor) \rceil \le 2+a/\sqrt{k}$ times. The first term on the left-hand side corresponds to a match of remaining requests at time $1$.
    Thus, the cost of $\ALG^*$ for $\sigma_{\ALG}$ is at most $2 + a/\sqrt{k}+\sqrt{k}$. 
    Hence the competitive ratio of $\ALG$ is at least
    \begin{align*}
        \sup_{\sigma}\frac{\ALG(\sigma)}{\OPT(\sigma)}
        \ge \frac{\ALG(\sigma_{\ALG})}{\ALG^*(\sigma_{\ALG})}
        = \frac{a + (k - 1)}{2 + a/\sqrt{k} + \sqrt{k}} = \sqrt{k} \cdot \left(1  - \frac{1+2\sqrt{k}}{a + k + 2\sqrt{k}} \right)=\Omega(\sqrt{k}).
    \end{align*}
    Therefore, the competitive ratio of $\ALG$ is $\Omega(\sqrt{k})$.
\end{proof}
\end{full}

In \Cref{subsec:upper}, we provide an $O(\log k/\log\log k)$-competitive algorithm by utilizing partial matchings of remaining requests.
Thus, every match-all-remaining algorithm is not optimal.
Subsequently, in \Cref{subsec:lower}, we establish the lower bound of $\Omega(\log k/\log\log k)$.

Before proceeding, we introduce some notations.
Recall that $\mathbb{Z}_k=\{0,1,\dots,k-1\}$.
\begin{definition}
We write $\overline{x} \in \mathbb{Z}_k$ to represent the remainder of $x \in \mathbb{Z}$ when divided by $k$.
With this notation, $a \equiv b \pmod k$ can be expressed $\overline{a} = \overline{b}$.
Additionally, for $\ell, r \in \mathbb{Z}_+$, we define the cyclic interval
\[
    \cint{\ell, r} \coloneqq \begin{cases}
        \{ x \in \mathbb{Z}_k \mid \overline{\ell} \leq x \leq \overline{r}\} 
        =\{\overline{\ell},\overline{\ell}+1,\dots,\overline{r}\}
        & (\text{if } \overline{\ell} \leq \overline{r}), \\
        \{ x \in \mathbb{Z}_k \mid x \leq \overline{r} \text{ or } \overline{\ell} \leq x\} 
        =\{\overline{\ell},\overline{\ell}+1,\dots,k-1,0,1,\dots,\overline{r}\}
        & (\text{if } \overline{\ell} > \overline{r}).
    \end{cases}
\]
The number of elements in an interval $\cint{\ell,r}$ is denoted by $|\cint{\ell, r}|$.
\end{definition}

\subsection{Upper Bound} \label{subsec:upper}

The objective of this subsection is to prove the following theorem.

\begin{restatable}{theorem}{thmupper} \label{thm:upper}
    For the OMDSC problem where the penalty function $f$ satisfies $B_f = \{kn \mid n \in \mathbb{Z}_{++}\}$, there exists a deterministic online algorithm with a competitive ratio of $O(\log k / \log \log k)$.
\end{restatable}

By \Cref{thm:match_all_remaining}, an $O(\log k/\log\log k)$-competitive algorithm must take into account both the timing and the size of requests to be matched.
In response to this requirement, we propose an algorithm that accounts for the costs of other algorithms.
The execution of the algorithm is organized into phases.
In each phase, our algorithm aims to increase the cost incurred by any algorithm to at least $1$, while ensuring that its own cost remains at most $O(\log k/\log\log k)$. 
We classify the competing algorithms based on the number of unmatched requests carried over from the previous phase.
As for the current phase, we can assume that algorithms only perform matches of size that are multiples of $k$, since otherwise the cost becomes at least $1$ immediately.
Furthermore, it is sufficient to consider ``greedy'' algorithms that immediately match $k$ requests as long as there are at least $k$ unmatched requests exist.

The first phase starts at the beginning.
Our algorithm runs with the following variables.
\begin{definition} \label{def:upper_variables}
At each time within each phase, the variables $W_0,W_1,\dots,W_{k-1}$, $s$, and $a$ are defined as follows:
\begin{itemize}
    \item $W_i$ $(i\in\mathbb{Z}_k)$: the waiting cost incurred thus far within the phase on the algorithm that greedily performs matches with no size cost, assuming $(k - i)$ unmatched requests carried over to the phase.
    
    \item $s$: the number of requests given thus far in the phase.

    \item $a$: the number of requests the proposed algorithm matches thus far during the phase.
\end{itemize}
\end{definition}

Note that, assuming $(k - i)$ unmatched requests carried over to the phase, the optimum cost incurred thus far within the phase is at least $\min\{1,\,W_i\}$.
At the end of the phase, our algorithm ensures $W_i \geq 1$ for all $i \in \mathbb{Z}_k$.
This implies that any algorithm incurs a cost of at least $1$ per phase, as matching a number of requests that is not a multiple of $k$ costs $1$.

During each phase, our algorithm recursively processes a subroutine with parameters of a cyclic interval $\cint{p, q}$ and an integer $l$, ensuring that $W_i\ge 1$ for all $i\not\in \cint{p,q}$ and $W_i\ge l/\alpha$ for all $i\in \cint{p,q}$, where $\alpha$ is a real such that $\alpha^\alpha = k$. Note that $\alpha = \Theta(\log k / \log \log k)$.%
\footnote{Given that $\alpha^\alpha = k$, the ratio $\alpha / (\log k / \log \log k) = \alpha / (\log \alpha^\alpha / \log \log \alpha^\alpha) = \alpha \cdot (\log \alpha + \log \log \alpha) / (\alpha \log \alpha)$ approaches $1$ as $\alpha\to\infty$. Hence, $\alpha=\Theta(\log k/\log\log k)$.}
In each iteration of the subroutine, the interval size $\big|\cint{p,q}\big|$ is reduced by a factor of $1/\alpha$ or $l$ is increased by one, with incurring cost at most $O(1)$.
Consequently, almost all $W_i$ are at least $1$ after $\Theta(\alpha)$ recursions.

At the beginning of each phase, $W_0,W_1,\dots,W_{k-1}$, $s$, and $a$ are initialized to $0$ and $\textbf{recurring}(\cint{0, k-1},\, 0)$ is called.
Throughout the phase, the algorithm updates the values of $W_0,W_1,\dots,W_{k-1}$, $s$, and $a$, appropriately.
Moreover, if at least $k$ unmatched requests exist (i.e., $s - a \geq k$), it immediately matches $k$ of them.

\begin{algorithm}[t]
\caption{The pseudo-code of $\textbf{recurring}(\cint{p, q}, l)$}\label{alg:multiplek}
\SetKwProg{Def}{def}{:}{}
\SetKwInput{Req}{Require}
\tcc{$W_0,W_1,\dots,W_{k-1}$, $s$, $a$ are variables that change as defined in \cref{def:upper_variables}}
\Def{$\mathbf{recurring}(\cint{p, q}, l)$}{
    \tcc{It is guaranteed that $p,q\in\mathbb{Z}_k$, $l\in\{0,1,\dots,\lceil\alpha\rceil\}$, $W_i\ge 1~(\forall i\notin\cint{p,q})$, $W_i\ge l/\alpha~(\forall i\in\cint{p,q})$, and $\overline{a}=p$}
    \lIf{$\big|\cint{p,q}\big|\le\alpha$ or $l\ge \alpha$}{\textbf{Exit} the recursion}\label{line:d1}
    \textbf{Wait} until $W_{\overline{a}}$ increases by $2$\tcp*{Cost:\,2}\label{line:p1}
    \lIf{$W_i\ge \frac{l+1}{\alpha}$ for all $i\in\cint{p,q}$}{\textbf{call} $\textbf{recurring}(\cint{p, q}, l+1)$}\label{line:d2}
    Let $\cint{p', q'}$ be the minimum cyclic interval where $W_i \geq \frac{l + 1}{\alpha}$ for all $i \notin \cint{p', q'}$\;\label{line:p2}
    \textbf{Wait} until $\overline{s} \in \cint{p', p - 1}$ or $W_p$ increases by $1$\tcp*{Cost:\,1}\label{line:p3}
    \lIf{$W_p$ increased by $1$ at line~$\ref{line:p2}$}{\textbf{call} $\textbf{recurring}(\cint{p, q}, l+1)$}\label{line:d3}
    \textbf{Match} $\overline{p'-a}$ requests\tcp*{Cost:\,1}\label{line:p4}
    \textbf{Wait} until $W_{\overline{a}}$ increases by $2$\tcp*{Cost:\,2}\label{line:p5}
    \If{$W_i\ge\frac{l+1}{\alpha}$ for all $i\in\cint{p',q'}$\label{line:c}}{
        \textbf{Wait} until $\overline{s}\in\cint{p,p'-1}$ or $W_{p'}$ increases by $1$\tcp*{Cost:\,$\le1$}\label{line:p6}
        \lIf(\tcp*[f]{Cost:\,$\le 1$}){$\overline{s}\in\cint{p,p'-1}$}{\textbf{match} $\overline{p-a}$ requests}\label{line:p7}
        \textbf{Call} $\textbf{recurring}(\cint{a, q}, l+1)$\;\label{line:d4}
    }
    \Else{
        Let $r\gets \argmin_{j\in\{p'+\lfloor 2L/\alpha\rfloor,\, q\}} \overline{j-p'}$\;\label{line:p8} 
        \textbf{Call} $\textbf{recurring}(\cint{p', r}, l)$\;\label{line:d5}
    }
}
\end{algorithm}

The subroutine of $\textbf{recurring}(\cint{p, q}, l)$ is defined as follows. 
See \Cref{alg:multiplek} for a formal description.
When this function is called, it is guaranteed that $p,q\in\mathbb{Z}_k$, $l\in\{0,1,\dots,\lceil\alpha\rceil\}$, $W_i \geq 1$ for all $i \notin \cint{p, q}$, $W_i \geq l / \alpha$ for all $i \in \cint{p, q}$, and $\overline{a} = p$ (recall that $\overline{a}\in\mathbb{Z}_k$ is the remainder of $a$ when divided by $k$). 
The recursion ends (line~\ref{line:d1}) when $\big|\cint{p, q}\big|$ is at most $\alpha$ (i.e., $W_i\ge 1$ for all $i\notin \cint{a,\,a+\lfloor\alpha\rfloor}$) or when $l$ is at least $\alpha$ (i.e., $W_i \geq 1$ for all $i \in \mathbb{Z}_k$).

At the beginning of each recursion, the algorithm waits until $W_{\overline{a}}$ increases by $2$ (line~\ref{line:p1}).
If $W_i\ge (l+1)/\alpha$ for all $i \in \cint{p,q}$ at the end of the wait, then $\textbf{recurring}(\cint{p, q}, l+1)$ is called (line~\ref{line:d2}).
Otherwise, $W_i$ increases by $1 / \alpha$ 
with the exception of those values of $i$ belonging to an interval $\cint{p', q'}\subseteq \cint{p, q}$ (line~\ref{line:p2}). 
In the latter case, the algorithm attempts to increase $W_i$ by $1/\alpha$ for all $i \in \cint{p', q'}$ or increase $W_i$ by $1$ except for $i$ belonging to some smaller interval within $\cint{p, q}$. 
Next, it tries to change the value of $\overline{a}$ to $p'$ by matching $\overline{p'-a}$ requests. To do this, it waits until $\overline{s}\in \cint{p', p-1}~(=\mathbb{Z}_k\setminus \cint{p,p'-1})$, but it stops waiting when $W_p$ increases by $1$ to prevent a high waiting cost (line~\ref{line:p3}, see also \cref{subfig:p3}).
If $W_p$ increases by $1$ while $\overline{s}$ is in $\cint{p, p'-1}$, $W_i$ increases by $1~(\geq 1 / \alpha)$ for all $i \in \cint{p', q'}$ and it calls $\textbf{recurring}(\cint{p, q}, l+1)$ (line~\ref{line:d3}).
Otherwise, it changes the value $\overline{a}$ to $p'$ by matching $\overline{p'-a}~(=\overline{p'-p})$ requests. 
Then, it waits until $W_{\overline{a}}~(=W_{p'})$ increases by $2$. 

When $W_{\overline{a}}$ increases by $2$, there are two cases to be considered. The first case is that $W_i\ge (l+1)/\alpha$ for all $i \in \cint{p', q'}$. Then, the algorithm attempts to call $\textbf{recurring}({\cint{p, q}}, l+1)$, but $\overline{a}$ must be set to $p$ before calling it to satisfy the condition of $\textbf{recurring}$. Therefore, the algorithm changes the value $\overline{a}$ to $p$ or increases $W_i$ by $1$ for all $i \in \cint{p, p'-1}$ in a similar way as above (line~\ref{line:p6}, see also \cref{subfig:p6}). 
Then, the algorithm calls $\textbf{recurring}({\cint{a, q}}, l+1)$ (line~\ref{line:d4}). 
The second case is that $W_i<(l+1)/\alpha$ for some $i \in \cint{p', q'}$. Then, $W_i$ increases by $1$ except for $i$ in a small interval $\cint{p', r}$. We will show that we can take $r=\argmin\nolimits_{j \in \{p' + \lfloor2L / \alpha\rfloor, q\}} \overline{j - p'}$. 
Thus, the algorithm calls $\textbf{recurring}({\cint{p', r}}, l)$ (line~\ref{line:d5}).

\begin{figure}[ht]
    \centering
    \begin{minipage}{.47\textwidth}
\centering
\begin{tikzpicture}[scale=.6,font=\small]
    \coordinate (p) at (1,0);
    \coordinate (q) at (9,0);
    \coordinate (pp) at (4,0);
    \coordinate (qq) at (8,0);

    \draw[red,thick,fill=red!10] (10,0) -- (pp) -- ++(0,0.3) -- (10,0.3) node[midway, above] {$\cint{p',p-1}$};
    \draw[red,thick,fill=red!10] (0, 0.3) -- ($(p)+(-0.5,0.3)$) -- (p) -- (0,0);

    \fill (p) circle (3pt) node[below] {$p=\overline{a}$};
    \fill (q) circle (3pt) node[below] {$q$};
    \fill (pp) circle (3pt) node[below] {$p'$};
    \fill (qq) circle (3pt) node[below] {$q'$};

    \draw[very thick] (-0.1,0) -- (10.1,0);
    \foreach \x in {0,0.5,...,10} \draw[very thin] (\x,.1) -- (\x,-.1);
    \node[below] at (10.3,0) {$\mathbb{Z}_k$};
\end{tikzpicture}
\subcaption{Situation at line \ref{line:p3}}\label{subfig:p3}
\end{minipage}%
\begin{minipage}{.06\textwidth}\mbox{}\end{minipage}%
\begin{minipage}{.47\textwidth}
\centering
\begin{tikzpicture}[scale=.6,font=\small]
    \coordinate (p) at (1,0);
    \coordinate (q) at (9,0);
    \coordinate (pp) at (4,0);
    \coordinate (qq) at (8,0);

    \draw[red,thick,fill=red!10] (p) -- ++(0,0.3) -- ($(pp)+(-.5,.3)$) node[midway, above] {$\cint{p,p'-1}$} -- (pp);

    \fill (p) circle (3pt) node[below] {$p$};
    \fill (q) circle (3pt) node[below] {$q$};
    \fill (pp) circle (3pt) node[below] {$p'=\overline{a}$};
    \fill (qq) circle (3pt) node[below] {$q'$};

    \draw[very thick] (-0.1,0) -- (10.1,0);
    \foreach \x in {0,0.5,...,10} \draw[very thin] (\x,.1) -- (\x,-.1);
    \node[below] at (10.3,0) {$\mathbb{Z}_k$};
\end{tikzpicture}
\subcaption{Situation at line \ref{line:p6}}\label{subfig:p6}
\end{minipage}%
    \caption{Relative positions of $p$, $q$, $p'$, $q'$, and $\overline{a}$ in \cref{alg:multiplek}}
    \label{fig:relations_s_p}
\end{figure}
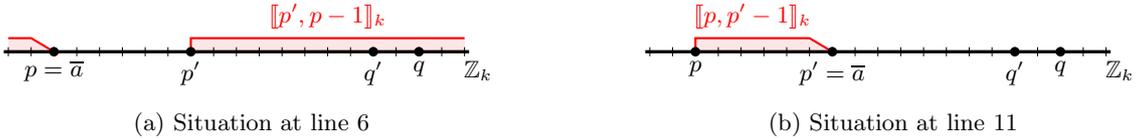

After completing the call of \textbf{recurring}, we have $W_i\ge 1$ for all $i\not\in \cint{a, a+\lfloor\alpha\rfloor}$.
The algorithm then executes the following steps to ensure that $W_i$ is at least $1$ for every $i\in\mathbb{Z}_k$.
\begin{enumerate}[Step 1]
    \item Wait until $W_{\overline{a}} \geq 1$ and $s > a$ hold, then match a single request. \label{processing_enu_1}
    \item If $W_i<1$ for some $i \in \mathbb{Z}_k$, go to Step~\ref{processing_enu_1}.
    \item Match all remaining requests and move to the next phase.
\end{enumerate}



To analyze the competitive ratio of this algorithm, we present several lemmas.\begin{short} Some proofs are deferred to \Cref{sec:proof_upper} due to space limitations.\end{short} We denote the values of variables $W_i$, $s$, and $a$ at time $t$ as $W_i(t)$, $s(t)$, and $a(t)$, respectively.
Here, the value of $W_i$ changes continuously within each phase, and its rate of increase over time is given by $\dv*{W_i(t)}{t} = \overline{s(t) - i}$. The increase rate of $W_i$ for each $i$ can be illustrated as in \cref{fig:increase}.

\begin{figure}[h]
    \centering
    \definecolor{darkgreen}{rgb}{0,0.3,0}
\begin{minipage}{.49\textwidth}
\centering
\begin{tikzpicture}[scale=.6,font=\small]
    \coordinate (s) at (3,0);
    \coordinate (i) at (7,0);
    \coordinate (p) at (1,0);
    \coordinate (q) at (9,0);
    \fill (s) circle (1pt) node[below,red] {$\overline{s}$};
    \fill (i) circle (1pt) node[blue,below] {$i$};
    \fill[blue] (p) circle (1pt) node[below] {$p$};
    \fill (q) circle (1pt) node[below] {$q$};
    \node[blue] at ($(p)+(0,1.7)$) {$\overline{s-p}$};
    \node[blue] at ($(i)+(0,1.7)$) {$\overline{s-i}$};
    \foreach \x in {0,0.5,...,10} \draw[very thin] (\x,.1) -- (\x,-.1);
    \foreach \x [count=\j from 0] in {0,0.5,...,10} {
        \pgfmathsetmacro{\i}{mod(10.5-\x+3,10.5)}
        \def\c{darkgreen};
        \ifnum\j=4\def\c{blue}\fi
        \ifnum\j=13\def\c{blue}\fi
        \draw[thick, \c, fill=\c!10] (\i-0.1,0) -- (\i-0.1,\j/12) -- ++(0.2,0) -- (\i+0.1,0);
    }
    \draw[very thick] (-0.2,0) -- (10.2,0);
    \node[below] at (10.3,0) {$\mathbb{Z}_k$};
\end{tikzpicture}
\subcaption{$\overline{s}\in\cint{p,i-1}$}\label{subfig:increaseA}
\end{minipage}%
\begin{minipage}{.02\textwidth}\mbox{}\end{minipage}%
\begin{minipage}{.49\textwidth}
\centering
\begin{tikzpicture}[scale=.6,font=\small]
    \def\s{8}
    \coordinate (s) at (\s,0);
    \coordinate (i) at (7,0);
    \coordinate (p) at (1,0);
    \coordinate (q) at (9,0);
    \fill (s) circle (1pt) node[below,red] {$\overline{s}$};
    \fill (i) circle (1pt) node[blue,below] {$i$};
    \fill[blue] (p) circle (1pt) node[below] {$p$};
    \fill (q) circle (1pt) node[below] {$q$};
    \node[blue] at ($(p)+(0,1.7)$) {$\overline{s-p}$};
    \node[blue] at ($(i)+(0,1.7)$) {$\overline{s-i}$};
    \foreach \x in {0,0.5,...,10} \draw[very thin] (\x,.1) -- (\x,-.1);
    \foreach \x [count=\j from 0] in {0,0.5,...,10} {
        \pgfmathsetmacro{\i}{mod(10.5-\x+\s,10.5)};
        \def\c{darkgreen};
        \ifnum\j=2\def\c{blue}\fi
        \ifnum\j=14\def\c{blue}\fi
        \draw[thick, \c, fill=\c!10] (\i-0.1,0) -- (\i-0.1,\j/12) -- ++(0.2,0) -- (\i+0.1,0);
    }
    \draw[very thick] (-0.2,0) -- (10.2,0);
    \node[below] at (10.3,0) {$\mathbb{Z}_k$};
\end{tikzpicture}
\subcaption{$\overline{s}\in\cint{i,p-1}$}\label{subfig:increaseB}
\end{minipage}
    \caption{The increase rate of $W_i$}
    \label{fig:increase}
\end{figure}
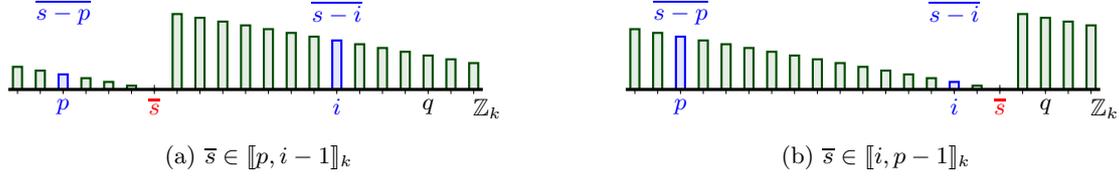

We first prove that, for $i\in\cint{p,q}$, if $W_p$ increases significantly while $W_i$ increases only slightly, then $W_p$ increases significantly during periods when $\overline{s-i}$ is small.
Specifically, $W_p$ increases at least $\alpha$ times faster than $W_i$ (i.e., $\dv*{W_p(t)}{t}\ge \alpha\cdot\dv*{W_i(t)}{t}$) at time $t$ when $\overline{s(t)}\in\cint{i,p-1}$ and 
$\overline{s(t)-i}\le \frac{1}{\alpha-1}\cdot(\overline{i-p})$ hold.
Taking these things into account, we can obtain the following lemma.

\begin{restatable}{lemma}{lmmsivalwait} \label{lmm:s_ival_wait_2}
    Let $p,q\in\mathbb{Z}_k$ and $i\in \cint{p,q}$. For two times $\tau_0,\tau_1~(\tau_0<\tau_1)$ in the same phase, suppose that $W_p(\tau_1)-W_p(\tau_0)\ge 2$ and $W_i(\tau_1)-W_i(\tau_0)< 1/\alpha$.
    Let $\gamma=\min\big\{\overline{p-i},\,\big\lceil (\overline{i - p}) / (\alpha - 1) \big\rceil\}$.
    Then, the increment of $W_p$ in time $t$ with $\tau_0\le t\le \tau_1$ and $\overline{s(t)} \in \cint{i,\, i + \gamma- 1}$ is more than $1$.
\end{restatable}

\begin{full}
    \begin{proof}
    Let $T\coloneqq\big\{t\mid \tau_0\le t\le\tau_1~\text{and}~\overline{s(t)} \in \cint{i, i + \gamma - 1}\big\}$ and  $\hat{T}\coloneqq\big\{t\mid \tau_0\le t\le\tau_1~\text{and}~\overline{s(t)} \not\in \cint{i, i + \gamma - 1}\big\}$.
    Suppose to the contrary that the increment of $W_p$ within $t\in T$ is at most $1$.
    Under this condition, $W_p$ increases by at least $1$ within $t\in \hat{T}$. 
    
    When $\gamma = \overline{p - i}$, the condition $\overline{s(t)} \notin =\cint{i, i + \gamma - 1} = \cint{i, p - 1}$ is equivalent to $\overline{s(t)} \in \cint{p, i - 1}$. 
    In contrast, when $\gamma = \left\lceil \left(\,\overline{i - p}\,\right) / (\alpha - 1) \right\rceil$, the condition $\overline{s(t)} \notin \cint{i, i + \gamma - 1}$ is equivalent to $\overline{s(t)} \in \cint{p, i - 1}\cup \cint{i+\gamma,p-1}$. 

    We will observe that $\dv*{W_p(t)}{t} \leq \alpha\cdot \dv*{W_i(t)}{t}$.
    For any $t \in \hat{T}$ with $\overline{s(t)} \in \cint{p, i - 1}$, we have $\overline{s(t) - p} \leq \overline{s(t) - i}$, which implies $\dv*{W_p(t)}{t} \leq \dv*{W_i(t)}{t}\leq \alpha\cdot \dv*{W_i(t)}{t}$.
    For any $t \in \hat{T}$ with $\gamma = \left\lceil \left(\,\overline{i - p}\,\right) / (\alpha - 1) \right\rceil$ and $\overline{s(t)} \in \cint{i + \gamma, p - 1}$, we have $\overline{s(t) - i}\ge \overline{(i+\gamma)-i}=\gamma$.
    Thus, in this case, we have
    \begin{align*}
        \alpha\cdot\dv{W_i(t)}{t}
        &=\alpha\cdot(\overline{s(t)-i})
        =(\overline{s(t)-i})+(\alpha-1)\cdot(\overline{s(t)-i})
        \ge (\overline{s(t)-i})+(\alpha-1)\cdot\gamma\\
        &= (\overline{s(t)-i})+(\alpha-1)\cdot\left\lceil\tfrac{(\overline{i-p})}{\alpha-1}\right\rceil
        \ge (\overline{s(t)-i})+(\overline{i-p})
        = \overline{s(t)-p}
        =\dv{W_p(t)}{t}.
    \end{align*}

    Therefore, we have
    \begin{align*}
    \frac{1}{\alpha} \cdot 1 \le \frac{1}{\alpha} \cdot \int_{\hat{T}} \dv{W_p(t)}{t}\dd{t} \leq \int_{\hat{T}} \dv{W_i(t)}{t} \dd{t}\le \int_{\tau_0}^{\tau_1}\dv{W_i(t)}{t} \dd{t}=W_i(\tau_1)-W_i(\tau_0),
    \end{align*}
    where the second inequality holds by $\dv*{W_p(t)}{t} \leq \alpha\cdot \dv*{W_i(t)}{t}$.
    This contradicts the assumption that $W_i(\tau_1)-W_i(\tau_0)< 1/\alpha$.
\end{proof}
\end{full}



Next, by using \Cref{lmm:s_ival_wait_2}, we show that $W_i$ increases by at least $1 / \alpha$ except for some small interval after $W_p$ increases by $2$.
\begin{restatable}{lemma}{lmmoptcostwait} \label{lmm:optcost_wait_2}
    Let $p, q \in \mathbb{Z}_k$. For two times $\tau_0, \tau_1~(\tau_0 < \tau_1)$ in the same phase, suppose that $W_p(\tau_1) - W_p(\tau_0) = 2$. Then, there exist some $p',q'\in\mathbb{Z}_k$ 
    such that $\cint{p',q'}\subseteq \cint{p,q}$, $\big|\cint{p', q'}\big|\le \big\lceil \big|\cint{p, q}\big| / \alpha \big\rceil$, and $W_i(\tau_1) - W_i(\tau_0) \geq 1 / \alpha$ for all $i \in \cint{p, q} \setminus \cint{p', q'}$.
\end{restatable}

\begin{full}
    \begin{proof}
If $W_i(\tau_1)-W_i(\tau_0)\ge 1/\alpha$ for all $i\in\cint{p,q}$, then the condition can be satisfied by setting $p'=q'=p$.
Hence, in what follows, we assume that $W_i(\tau_1)-W_i(\tau_0)< 1/\alpha$ for some $i\in\cint{p,q}$.
Define
\begin{align*}
p'&\in\argmin\nolimits_{i\in\cint{p,q}:\,W_i(\tau_1)-W_i(\tau_0)< 1/\alpha}\overline{i-p}
\ \ \text{and}\\ 
q'&\in\argmax\nolimits_{i\in\cint{p,q}:\,W_i(\tau_1)-W_i(\tau_0)< 1/\alpha}\overline{i-p}.
\end{align*}
Then, it is not difficult to see that $\cint{p',q'}\subseteq \cint{p,q}$ and $W_i(\tau_1) - W_i(\tau_0) \geq 1 / \alpha$ for all $i \in \cint{p, q} \setminus \cint{p', q'}$.

What is left is to show that $\big|\cint{p', q'}\big|\le \big\lceil \big|\cint{p, q}\big| / \alpha \big\rceil$.
Let $\gamma_i \coloneqq \min \big\{\overline{p - i},~\big\lceil\left(\overline{i - p}\right) / (\alpha - 1)\big\rceil\big\}$ and $T_i \coloneqq \{t \mid \tau_0 \leq t \leq \tau_1~\text{and}~\overline{s(t)} \in \cint{i, i + \gamma_i - 1}\}$ for $i\in\{p',q'\}$.
Then, by \Cref{lmm:s_ival_wait_2}, the increments of $W_p$ within $t \in T_{p'}$ and $t\in T_{q'}$ are more than $1$, respectively.
Since the total increment of $W_p$ is $2$, the intervals $\cint{p', p' + \gamma_{p'} - 1}$ and $\cint{q', q' + \gamma_{q'} - 1}$ must overlap. 

If $\overline{p' - p} \leq (\alpha - 1) \cdot \big|\cint{p, q}\big| / \alpha$, then
$\gamma_{p'} \leq \big\lceil (\overline{p' - p})/(\alpha - 1)\big\rceil \leq \big\lceil \big|\cint{p, q}\big|/\alpha\big\rceil$.
From $\cint{p', p' + \gamma_{p'} - 1}\cap \cint{q', q' + \gamma_{q'} - 1}\ne\emptyset$ and $\overline{p' - p} \le \overline{q' - p}$, we get $q' \in \cint{p', p' + \gamma_{p'} - 1}$. This implies that $\big|\cint{p',q'}\big|=(\overline{q' - p'}) + 1 \leq \gamma_{p'} \leq \big\lceil \big|\cint{p, q}\big| / \alpha \big\rceil$.
On the other hand, if $\overline{p' - p} > (\alpha - 1) \cdot \big|\cint{p, q}\big| / \alpha$, we obtain
\begin{align*}
    \big|\cint{p',q'}\big|
    &\le \big|\cint{p',q}\big|
    = \big|\cint{p,q}\big| - \big|\cint{p,p'-1}\big|
    = \big|\cint{p,q}\big| - \overline{p'-p}\\
    &< \big|\cint{p,q}\big| - (\alpha - 1) \cdot \big|\cint{p, q}\big| / \alpha
    = \big|\cint{p, q}\big| / \alpha
    \le \big\lceil \big|\cint{p, q}\big| / \alpha \big\rceil. \qedhere
\end{align*}
\end{proof}
\end{full}

Now, we demonstrate that $\textbf{recurring}$ satisfies three desired conditions.
\begin{restatable}{lemma}{lmmreccondition} \label{lmm:rec_condition}
    When $\textbf{recurring}(\cint{p, q}, l)$ is called, the following three conditions are satisfied:
    $(i)$ $W_i \geq 1$ for all $i \notin \cint{p, q}$, 
    $(ii)$ $W_i \geq l / \alpha$ for all $i \in \cint{p, q}$, and
    $(iii)$ $\overline{a} = p$.
\end{restatable}

\begin{proof}
We prove this by induction. 

Initially, when $\textbf{recurring}(\cint{0, k-1}, 0)$ is called at the beginning of a phase, the conditions are satisfied as $W_i=0\ge 0/\alpha$ for all $i\in \cint{0,k-1}$ and $\overline{a}=p=0$.

As an induction hypothesis, suppose the conditions are satisfied when $\textbf{recurring}(\cint{p, q}, l)$ is called. Our goal is to show that these conditions continue to be satisfied at the next recursive call.
We analyze this based on where the recursive call is made in the conditional branching of \Cref{alg:multiplek}. We denote by $L$ the cardinality of $\cint{p, q}$.

\begin{description}
\setlength{\itemsep}{6pt}
\item[If the condition of line~\ref{line:d2} is true and ${\textbf{recurring}(\cint{p, q}, l + 1)}$ is called] 
    Conditions (i) and (iii) are satisfied as $\overline{a} = p$. Condition (ii) is also satisfied since the condition of line~\ref{line:d2} is true. 

\item[If the condition of line~\ref{line:d2} is false] 
    Since the proposed algorithm waits until $W_{\overline{a}}$ increases by $2$ at line~\ref{line:p1}, the interval $\cint{p', q'}$ chosen at line~\ref{line:p2} satisfies $\big|\cint{p', q'} \big| \leq \lceil L / \alpha \rceil$ by \Cref{lmm:optcost_wait_2}.

\item[If the condition of line~\ref{line:d3} is true and ${\textbf{recurring}(\cint{p, q}, l + 1)}$ is called] 
    Since $W_p$ increases by $1$ while $\overline{s} \in \cint{p, p' - 1}$,  
    $W_i$ also increases by at least $1$ for all $i \in \cint{p', p - 1}$ because $\overline{s-i}\ge \overline{s-p}$. 
    Thus, conditions (i) and (iii) are satisfied as $\overline{a} = p$. Additionally, for all $i \in \cint{p', q'}$, $W_i$ increases by $1 / \alpha$, and hence condition (ii) is also satisfied.

\item[If the condition of line~\ref{line:d3} is false] 
    Since $\overline{s} \in \cint{p', p-1}$, the proposed algorithm can match $(\overline{p' - a})=(\overline{p'-p})$ requests at line~\ref{line:p4}. This match results in $\overline{a} = p'$. 

\item[If the condition of line~\ref{line:c} is true and ${\textbf{recurring}(\cint{a, q}, l + 1)}$ is called] 
    at line~\ref{line:p7}, if $\overline{s} \in \cint{p, p' - 1}$, the proposed algorithm can match $(\overline{p - a})=(\overline{p-p'})$ requests, making $\overline{a} = p$.
    On the other hand, if $W_{p'}$ increases by $1$ at line~\ref{line:p6}, it indicates that $\overline{s} \in \cint{p', p - 1}$ during the increase. Thus, for any $i \in \cint{p, p' - 1}$, $W_i$ increases by at least $1$, and $\overline{a} = p'$.
    This satisfies conditions (i) and (iii) in both cases, and condition (ii) is satisfied since the condition of line~\ref{line:c} is true.

\item[If the condition of line~\ref{line:c} is false and ${\textbf{recurring}(\cint{p', r}, l)}$ is called] 
    There exists an index $i \in \cint{p', q'}$ such that $W_i$ increases by less than $1 / \alpha$ at line~\ref{line:p5}.
    With $\overline{a} = p'$ and letting $\gamma = \min \big\{\overline{p' - i},~\big\lceil \left(\overline{i - p'}\right) / (\alpha - 1) \big\rceil\big\}$, \Cref{lmm:s_ival_wait_2} ensures that $W_{\overline{a}}$ increases by more than $1$ during $\overline{s} \in \cint{i, i + \gamma - 1}$.
    Therefore, for all $j \in \cint{i + \gamma, p' - 1}$, $W_j$ also increases by at least $1$.
    By $(\overline{i-p'})\le (\overline{q'-p'})\le \big\lceil \cint{p',q'}\big\rceil-1\le \lceil L/\alpha\rceil-1\le \lfloor L/\alpha\rfloor$,
    we obtain
    \begin{align*}
        \cint{p',\, i + \gamma - 1} 
        &\subseteq \bigcint{ p',\, i + \big\lceil (\overline{i-p'})/(\alpha-1)\big\rceil -1 } \\
        &\subseteq \bigcint{ p',\, p' + (\overline{i - p'}) + \big\lceil (\overline{i-p'})/(\alpha-1)\big\rceil -1 } \\
        &\subseteq \bigcint{ p',\, p' + 2(\overline{i - p'})} 
        \subseteq \bigcint{ p',\, p' + 2\lfloor L/\alpha\rfloor}
        \subseteq \bigcint{ p',\, p' + \lfloor 2L/\alpha\rfloor}.
    \end{align*}
    Thus, by setting $r \coloneqq \argmin\nolimits_{j \in \{p' + \lfloor2L / \alpha\rfloor, q\}} \overline{j - p'}$, $W_j$ increases by at least $1$ for all $j \in \cint{p, q} \setminus \cint{p', r}$.
    Since $\overline{a} = p'$, conditions (i) and (iii) are satisfied, and by the induction hypothesis, condition (ii) is also satisfied. \qedhere
\end{description}
\end{proof}

We can bound the number of recursions called by the proposed algorithm from \Cref{lmm:rec_condition}, giving us the upper bound of the cost incurred by the proposed algorithm during each phase.

\begin{restatable}{lemma}{lmmalgcostphase} \label{lmm:algcost_phase}
    The proposed algorithm incurs a cost of $O(\alpha)$ per phase.
\end{restatable}

\begin{full}
    \begin{proof}
    By the definition of \Cref{alg:multiplek}, each call of $\textbf{recurring}(\cint{p, q}, l)$ either increments $l$ by $1$ or reduces the number of elements in $\cint{p, q}$ to at most $\big\lfloor 2 \cdot \big|\cint{p, q}\big| / \alpha \big\rfloor + 1$.
    Since the recursion ends when $L \leq \alpha$ or $l \geq \alpha$, the addition to $l$ occurs at most $\lceil\alpha\rceil$ times.
    Let $L_n$ be the number of elements after the $n$th reduction. 
    Then, for $n$ with $L_n>\alpha$, we have
    \[
        L_{n} \leq \frac{2}{\alpha} L_{n-1} + 1 \leq \frac{3}{\alpha} L_{n-1} \leq \left(\frac{3}{\alpha}\right)^{n} \cdot k = \left(\frac{3}{\alpha}\right)^{n} \cdot \alpha^\alpha,
    \]
    by $L_0 = k$.
    The number of elements in the interval after $2\lceil\alpha\rceil$ iterations is at most
    \[
        L_{2\lceil\alpha\rceil} \leq \left(\frac{3}{\alpha}\right)^{2\alpha} \cdot \alpha^\alpha \leq \frac{3^{2\alpha}}{\alpha^\alpha} 
        \leq \left(\frac{9}{\alpha}\right)^{\alpha} 
        \leq \alpha,
    \]
    where the last inequality is justified at least when $\alpha \geq 9$.
    Consequently, the number of elements is reduced to $\alpha$ or less in $O(\alpha)$ iterations. Thus, the number of recursive calls is $O(\alpha)$.

    The cost incurred by the algorithm in each recursive call can be determined by examining the critical path of maximum cost incurred at each line in the process,
    which is no more than $8$. Combined with $O(\alpha)$ recursive calls, the total cost is also $O(\alpha)$.

    After recursion ends, the algorithm attempts to ensure $W_i \geq 1$ for all $i \in \mathbb{Z}_k$. 
    If the recursion finishes by $l\ge \alpha$, then the value $W_i$ is at least $l/\alpha=1$ for all $i\in\mathbb{Z}_k$. 
    Conversely, if the recursion finishes by $L\le \alpha$, each element in this interval can be treated at a cost of at most $2$ ($1$ for the size cost and $1$ for the waiting cost). Consequently, it is possible to achieve $W_i\ge 1$ for all $i\in\mathbb{Z}_k$ at a cost of at most $2L=O(\alpha)$.
    Completing the phase may require an additional cost of $1$ to match all remaining requests.

    Therefore, the cost incurred by the proposed algorithm per phase is $O(\alpha)$.
\end{proof}
\end{full}

Next, we provide a lower bound on the cost incurred by any algorithm during each phase.
\begin{restatable}{lemma}{lmmoptcostphase} \label{lmm:optcost_phase}
    At the end of a phase, the cost incurred within that phase by any algorithm is at least $1$ (when we treat the waiting costs as imposed sequentially at each moment, rather than at the time of matching). 
\end{restatable}

\begin{full}
    \begin{proof}
    If an algorithm matches a number of requests that is not a multiple of $k$ within a phase, it immediately incurs a cost of $1$.
    Consequently, we only need to consider algorithms that do not match a set of requests whose size is not a multiple of $k$ within a phase. Furthermore, matching $k$ unmatched requests immediately does not incur a loss, so we consider such algorithms.

    An algorithm may carry over some requests from the previous phase. 
    For each $i \in \mathbb{Z}_k$, $W_i$ represents the waiting cost incurred by an algorithm that carries over $\overline{k - i}$ requests. This is because, when the number of requests already given in the phase is $s$, the number of requests left unmatched by an algorithm carrying over $\overline{k - i}$ requests is $\overline{s + k - i} = \overline{s - i}$, which equals the rate of increase of $W_i$ by the proposed algorithm.
    At the end of a phase, $W_i \geq 1$ for all $i \in \mathbb{Z}_k$. This implies that a waiting cost of at least $1$ is incurred regardless of the number of requests carried over.

    Therefore, any algorithm incurs a cost of at least $1$ within a phase.
\end{proof}

\end{full}

Based on the lemmas above, we now prove \Cref{thm:upper}.

\begin{proof} [Proof of Theorem \ref{thm:upper}]
    Suppose that the proposed algorithm completes $p~(\geq 1)$ phases for the given instance. 
    Then, the cost for the proposed algorithm is at most $(p + 1) \cdot O(\alpha)$ by \Cref{lmm:algcost_phase}, while the cost for any algorithm is at least $p$ by \Cref{lmm:optcost_phase}. Therefore, the competitive ratio for this instance is at most $(p + 1)\cdot O(\alpha) / p = O(\alpha)$.

    Conversely, suppose that the instance ends during the first phase. 
    If the optimal offline algorithm incurs a cost of at least $1$, then the cost incurred by the proposed algorithm during the first phase is $O(\alpha)$, resulting in the competitive ratio of $O(\alpha)$.
    If the cost incurred by the optimal offline algorithm is less than $1$, then only waiting costs are incurred as no requests are carried over to the first phase. This means that the total cost of the optimal offline algorithm is equal to $W_0$. Since the proposed algorithm continues to wait until $W_0$ reaches $2$, the total cost of the proposed algorithm is also equal to $W_0$. Therefore, the competitive ratio is $1$ for such an instance.

    Therefore, the proposed algorithm is $O(\alpha)~(=O(\log k / \log \log k))$-competitive.
\end{proof}





\subsection{Lower Bound} \label{subsec:lower}

In this subsection, we prove the following lower bound of the competitive ratio.
\begin{theorem} \label{thm:lower_bound}
For the OMDSC problem with a penalty function $f$ that satisfies $B_f = \{kn \mid n \in \mathbb{Z}_{++}\}$, 
every deterministic online algorithm has a competitive ratio of $\Omega(\log k / \log \log k)$.
\end{theorem}

Let $\alpha$ be a positive real such that $\alpha^\alpha = k$.
We fix an arbitrary online algorithm $\ALG$ and construct an instance according to the behavior of $\ALG$, denoted as $\sigma_{\ALG}$. The instance $\sigma_{\ALG}$ is composed of multiple rounds. To construct $\sigma_{\ALG}$, we introduce variables similar to \cref{def:upper_variables}.

\begin{definition} \label{def:lower_variables}
For the instance $\sigma_\ALG$ and time $t$, the variables $W_0(t),W_1(t),\dots,W_{k-1}(t)$, $s(t)$, and $a(t)$ are defined as follows:
\begin{itemize}
    \item $W_i(t)~(i \in \mathbb{Z}_k)$: the waiting cost incurred by time $t$ on the algorithm that greedily performs matches with no size cost, assuming $i$ requests are matched at time $0$.
    \item $s(t)$: the number of requests given by time $t$.
    \item $a(t)$: the number of requests matched by $\ALG$ by time $t$.
\end{itemize}
\end{definition}
We will omit the argument $t$ when there is no confusion.

The instance $\sigma_{\ALG}$ contains $k-1$ requests with arrival time at $0$.
For $\sigma_{\ALG}$ and any $i \in \mathbb{Z}_k$, there exists an offline algorithm that incurs a cost of at most $2+W_i(t^*)$, where $t^*$ is the time when the last requests are given.
This can be achieved by matching $i$ requests at time $0$, greedily matching a multiple of $k$ requests in the middle, and matching all the remaining requests at time $t^*$.
While constructing $\sigma_{\ALG}$, we manage an interval $\cint{p, q}$ where $W_i = O(1)$ for any $i \in \cint{p, q}$. As rounds progress, the interval is gradually narrowed, and $\ALG$ incurs a cost of $\Omega(1)$ per round.
We show that it is possible to construct an instance $\sigma_{\ALG}$ consisting of $\Omega(\alpha)$ rounds while keeping such an interval $\cint{p, q}$ nonempty.
This implies that while there exists an offline algorithm with a cost of $O(1)$, $\ALG$ incurs a cost of $\Omega(\alpha)$, indicating that the competitive ratio of $\ALG$ is $\Omega(\alpha)=\Omega(\log k/\log\log k)$.

The interval $\cint{p, q}$ is initialized to $\cint{0, k - 1}$. Additionally, we set variable $n$, which counts the number of rounds, to $0$.
The instance gives $k - 1$ requests at time $0$.
Subsequent requests are provided based on the following events, ensuring that (i) $k/\alpha^{2n}\le |\cint{p, q}|\le k / \alpha^n$, (ii) $W_i \leq 2n / \alpha$ for all $i \in \cint{p, q}$, and (iii) $\overline{a} \in \cint{q + 1, p}$ and $\overline{s} = q$.
If multiple events occur simultaneously, we assume that the event listed first takes precedence.

\begin{description}
\setlength{\itemsep}{6pt}
    \item[(a) {$\big|\cint{p, q}\big| < \alpha^2$}] 
        Set $n^*$ to be $n$ and finalize the instance.

    \item[(b) $\ALG$ matches a non-multiple of $k$ requests] 
        Let $h = \big\lceil |\cint{p, q}| / \alpha^2 \big\rceil$.
        If $\overline{a} \notin \cint{q - h + 1, q}$ after the match, then replace $(p, q)$ with $(\overline{q - h + 1}, q)$. Otherwise, replace $(p, q)$ with $(\overline{q - 2h + 1}, \overline{q - h})$.
        Then, give $\overline{q - s}$ requests so that $\overline{s} = q$. Proceed to the next round and increase round count $n$ by $1$. 

    \item[(c) $W_{\overline{a}}$ increased by $1$ from the previous event] 
        Let $h = \big\lceil |\cint{p, q}| / \alpha^2 \big\rceil$ and replace $(p,q)$ with $(\overline{q - h + 1},q)$. Then, proceed to the next round and increase round count $n$ by $1$.
\end{description}

\begin{figure}[ht]
    \centering
    \begin{tikzpicture}[scale=.8,font=\small]
    \coordinate (p) at (2,0);
    \coordinate (q) at (9,0);
    \coordinate (hh) at (6,0);
    \coordinate (h) at (7.5,0);

    \draw[red,thick,fill=red!10] (10,0) -- (q) -- ++(0.5,0.3) -- (10,0.3);
    \draw[red,thick,fill=red!10] (0, 0.3) -- ($(p)+(0,0.3)$) node[midway, above] {$\cint{q+1,p}$} -- (p) -- (0,0);

    \draw[blue,thick,fill=blue!10] (hh) -- ++(0.5,0.3) -- ($(h)+(0,0.3)$) -- (h);
    \draw[blue,thick,fill=blue!10] (h) -- ++(0.5,0.3) -- ($(q)+(0,0.3)$) -- (q);

    \fill (p) circle (3pt) node[below] {$p$};
    \fill (q) circle (3pt) node[below] {$q=\overline{s}$};
    \fill (h) circle (3pt) node[below] {$q{-}h$};
    \fill (hh) circle (3pt) node[below] {$q{-}2h$};

    \draw[very thick] (-0.1,0) -- (10.1,0);
    \foreach \x in {0,0.5,...,10} \draw[very thin] (\x,.1) -- (\x,-.1);
    \node[below] at (10.3,0) {$\mathbb{Z}_k$};
\end{tikzpicture}
    \caption{Relative positions of $p$, $q$, and $\overline{s}$}
    \label{fig:relations_lower}
\end{figure}
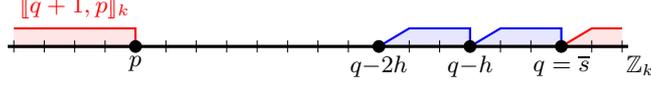

Now, we prove that the competitive ratio of $\ALG$ for $\sigma_{\ALG}$ is $\Omega(\alpha)$.
To prove the lower bound, we bound the cost incurred by any online algorithm $\ALG$ on $\sigma_{\ALG}$ from below and the cost for the optimal offline algorithm on $\sigma_{\ALG}$ from above.
\begin{short}
Due to space limitations, some proofs are deferred to \Cref{sec:proof_lower}.
\end{short}

\begin{full}
    First, we show a lemma that establishes a lower bound on the number of elements in an interval where the increment of $W_i$ in the interval is at most $2 / \alpha$ times the increment of $W_j$ for specific $j$.

\begin{restatable}{lemma}{lmmadvcondinc} \label{lmm:adv_cond_inc}
    For an interval $\cint{p,q}$ with the number of elements at least $\alpha^2$, let $L \coloneqq \big|\cint{p, q}\big|$ and $h \coloneqq \left\lceil L / \alpha^2 \right\rceil$.
    Then, for any $i \in \cint{q + 1, p}$ and $j \in \cint{q - 2h + 1, q}$, the increment of $W_j$ is at most $2 / \alpha$ times the increment of $W_i$ in the state where $\overline{s} = q$.
\end{restatable}

\begin{proof}
    By $\alpha \geq 4$ and $L \geq \alpha^2$, we have
    \begin{align}
        \frac{2}{\alpha} \cdot (L - 1) - \left( \frac{2L}{\alpha^2} + 1 \right)
        = \frac{2L}{\alpha} \cdot \left( 1 - \frac{1}{\alpha} \right) - \frac{2}{\alpha} - 1
        \geq \frac{2\alpha^2}{\alpha} \cdot \left(1 - \frac{1}{4}\right) - \frac{2}{4} - 1 > 0.
        \label{ineq:lmm_adv_cond_inc}
    \end{align}
    Therefore, for any $i \in \cint{q + 1, p}$ and $j\in \cint{q - 2h + 1, q}$, we get
    \begin{align*}
        \dv{W_j}{t} = \overline{q - j} 
        \leq 2 \cdot \left\lceil \frac{L}{\alpha^2} \right\rceil - 1
        \leq 2 \cdot \frac{L}{\alpha^2} + 1
        \leq \frac{2}{\alpha} \cdot \left(L - 1\right)
        \leq \frac{2}{\alpha} \cdot \overline{q - p}
        \leq \frac{2}{\alpha} \cdot \dv{W_i}{t},
    \end{align*}
    where the third inequality follows from \eqref{ineq:lmm_adv_cond_inc}.
    Since the rate of increase is constant while $\overline{s} = q$, the increment of $W_j$ is at most $2 / \alpha$ times the increment of $W_i$.
\end{proof}

By using this lemma, we bound the number of elements in $\cint{p, q}$ after $n$ rounds from below, and bound the values of $W_i$ for all $i \in \cint{p, q}$ from above.
\begin{restatable}{lemma}{lmmadvcondition} \label{lmm:adv_condition}
    Let $n\in\{0,1,2,\dots,n^*\}$.
    When the $n$th round of $\sigma_{\ALG}$ begins, the following three conditions are satisfied:
    (i) the number of elements in the interval $\cint{p, q}$ is at least $k / \alpha^{2n}$ and at most $k / \alpha^n$, (ii) $W_i \leq 2n / \alpha$ for all $i \in \cint{p, q}$, and (iii) $\overline{a} \in \cint{q + 1, p}$ and $\overline{s} = q$.
\end{restatable}

\begin{proof}
    We prove this by induction on $n$.

    At the beginning of $0$th round, the number of elements in $\cint{p, q}$ is $k / \alpha^0 = k$ since $\cint{p, q} = \cint{0, k - 1}$, and $W_i = 0 \leq 2 \cdot 0 / \alpha$ for all $i \in \cint{0, k - 1}$. 
    In addition, we have $\overline{a} = 0\in \cint{0,0}=\cint{k,0}$ and $\overline{s} = k - 1$.
    Thus, conditions (i), (ii), and (iii) are satisfied. 

    For a non-negative integer $n \in \{0,1,2,\dots,n^*-1\}$, suppose that conditions (i), (ii), and (iii) are satisfied at the beginning of the $n$th round. 
    We show that the three conditions are also satisfied at the beginning of the $(n+1)$st round. 
    There are two types of events that trigger a transition to the next round: (b) $\ALG$ matches a non-multiple of $k$ requests, and (c) $W_{\overline{a}}$ has increased by $1$ since the last event.
    In both cases, we prove that the inductive step can be carried out by using \Cref{lmm:adv_cond_inc}. 
    Let $p,q,W_i~(i\in\mathbb{Z}_k),s,a$ be the variables at the beginning of the $n$th round and $p',q',W'_i~(i\in\mathbb{Z}_k),s',a'$ be the variables at the beginning of the $(n+1)$st round.
    In addition, let $h=\big\lceil\big|\cint{p,q}\rceil/\alpha^2\big\rceil$.
    By the inductive hypothesis, we have $k/\alpha^{2n}\le \big|\cint{p,q}\big|\le k/\alpha^{n}$ and $\overline{a} \in \cint{q + 1, p}$.
    Also, as $n+1\le n^*$, we have $\big|\cint{p,q}\big|\ge \alpha^2$.

\begin{description}
\setlength{\itemsep}{6pt}
    \item[(b) When $\ALG$ matches a non-multiple of $k$ requests] \,\\
        By the definition of the procedure, we have $\big|\cint{p',q'}\big|=h=\big\lceil\big|\cint{p,q}\big|/\alpha^2\big\rceil\ge \lceil(k/\alpha^{2n})/\alpha^2\rceil\ge k/\alpha^{2n+2}$.
        Also, we have $\big|\cint{p',q'}\big|=h=\big\lceil\big|\cint{p,q}\big|/\alpha^2\big\rceil\le
        1+\big|\cint{p,q}\big|/\alpha^2\le \big|\cint{p,q}\big|/\alpha\le k/\alpha^{n+1}$ by $\big|\cint{p,q}\big|\ge \alpha^2$ and $\alpha\ge 4$.
        Thus, the condition (i) is satisfied.
        As $W_{\overline{a}}$ increases at most $1$ in the $n$th round, the increment of $W_j$ is less than $2 / \alpha$ for all $j \in \cint{p', q'} \subseteq \cint{q - 2h + 1, q}$ by \Cref{lmm:adv_cond_inc}.
        Therefore, $W'_j \leq 2(n + 1) / \alpha$ for all $j \in \cint{p', q'}$, satisfying the condition (ii).
        Moreover, we have $\overline{a'}\in\cint{q'+1,p'}$ and $\overline{s'}=q'$, and thus the condition (iii) is also met.

    \item[(c) When $W_{\overline{a}}$ has increased by $1$ since the last event] \,\\
        By definition, we have $p'=q-h+1$, $q'=q$, $s'=s$, and $a'=a$.
        Thus, we have $\big|\cint{p',q'}\big|=h$. Consequently, by applying the same analysis as above, we conclude that $k/\alpha^{2n}\le \big|\cint{p',q'}\big|\le k/\alpha^n$.
        This ensures that the condition (i) is satisfied
        As $W_{\overline{a}}$ increases exactly $1$ in the $n$th round, the increment of $W_j$ is less than $2 / \alpha$ for all $j \in \cint{p', q'} = \cint{q-h+1,q}\subseteq \cint{q-2h+1,q}$ by \Cref{lmm:adv_cond_inc}. Thus, $W'_j \leq 2(n + 1) / \alpha$ for all $j \in \cint{p', q'}$, satisfying the condition (ii).
        Furthermore, we have $\overline{a'}=\overline{a}\in \cint{q+1,p}\subseteq \cint{q'+1,p'}$ and $\overline{s'}=q'$, and thus the condition (iii) is also met.
\end{description}
Therefore, in either case, the conditions (i), (ii), and (iii) are satisfied at the beginning of the $(n+1)$st round. This completes the proof by induction.
\end{proof}

In what follows, we bound the cost incurred by any online algorithm $\ALG$ on $\sigma_{\ALG}$ from below and the cost for the optimal offline algorithm on $\sigma_{\ALG}$ from above.

\end{full}

\begin{restatable}{lemma}{lmmadvstepnumber}\label{lmm:adv_stepnumber}
    The number of rounds $n^*$ satisfies $\alpha / 2 - 1\le n^*\le \alpha$.
\end{restatable}

\begin{full}
    \begin{proof}
%
    For $n\in\{0,1,2,\dots,n^*\}$, let $L_n$ be the number of elements in the interval at the $n$th round. 
    From \Cref{lmm:adv_condition}, we have $k/\alpha^{2n}\le L_n\le k/\alpha^n$.
    %
    By the termination condition, we have $L_{n^*} < \alpha^2$ and $L_n\ge \alpha^2$ for all $n\in\{0,1,2,\dots,n^*-1\}$. 
    Since $\alpha^2>L_{n^*}\ge k/\alpha^{2n^*}=\alpha^{\alpha-2n^*}$, we have $n^*\ge (\alpha-2)/2 = \alpha/2 - 1$.
    In addition, since $\alpha^2\le L_{n^*-1}\le k/\alpha^{n^*-1}=\alpha^{\alpha-n^*+1}$, we have $n^*\le \alpha-1\le \alpha$.
    Therefore, we obtain $\alpha/2 - 1 \le n^*\le \alpha$.
\end{proof}
\end{full}

From \Cref{lmm:adv_stepnumber}, we can prove the following lemmas.
\begin{restatable}{lemma}{lmmoptupperforlb} \label{lmm:opt_upper_for_lb}
    For any deterministic online algorithm $\ALG$, there exists an offline algorithm that incurs a cost of at most $4$ for the instance $\sigma_{\ALG}$.
\end{restatable}

\begin{full}
    \begin{proof}
Let $[p,q]$ be the interval at the end of the instance, and let $i^*$ be an arbitrary element in $[p,q]$.
According to \Cref{lmm:adv_condition}, the waiting cost $W_i$ is at most $2n^*/\alpha$ when the instance construction ends.
Furthermore, by applying \Cref{lmm:adv_stepnumber}, we have $2n^*/\alpha\le 2\alpha/\alpha=2$.

Consider an algorithm that matches $\overline{i^*-1}$ requests at time $0$, subsequently matches requests only in multiples of $k$ and does so instantaneously, and finally matches all the remaining unmatched requests when the last request is presented.
Then, it is not difficult to see that the total waiting cost of this algorithm is at most $W_{i^*}$.
Moreover, the total size cost is at most $2$, as it matches at most two sets of requests of size non a multiple of $k$.
Therefore, The total cost incurred by this algorithm is at most $4$.
\end{proof}
\end{full}

\begin{restatable}{lemma}{lmmalglowerforlb} \label{lmm:alg_lower_for_lb}
    $\ALG(\sigma_{\ALG})\ge \alpha/2-1$ for any deterministic online algorithm $\ALG$.
\end{restatable}

\begin{full}
    \begin{proof}
    The number of rounds in $\sigma_{\ALG}$ is at least $\alpha/2-1$ by \Cref{lmm:adv_stepnumber}.
    As the rounds move on to the next, when the algorithm matches a non-multiple of $k$ requests or $W_{\overline{a}}$ increases by $1$, the algorithm incurs at least a cost of $1$ in each round.
    Therefore, the total cost incurred over the instance is at least $\alpha/2-1$.
\end{proof}
\end{full}

Now, we are ready to prove \Cref{thm:lower_bound}.
\begin{proof}[Proof of \Cref{thm:lower_bound}]
By combining \Cref{lmm:opt_upper_for_lb,lmm:alg_lower_for_lb}, the competitive ratio of $\ALG$ is at least 
\begin{equation*}
\sup_{\sigma} \frac{\ALG(\sigma)}{\OPT(\sigma)} 
\geq \frac{\ALG(\sigma_{\ALG})}{\OPT(\sigma_{\ALG})} 
\ge \frac{\alpha/2-1}{4} = \frac{\alpha}{8}-\frac{1}{4}
= \Omega\left(\frac{\log k}{\log \log k}\right).
\qedhere
\end{equation*}
\end{proof}

\section*{Acknowledgements}
This work was partially supported by JSPS KAKENHI Grant Number JP20K19739, JST PRESTO Grant Number JPMJPR2122, and Value Exchange Engineering, a joint research project between Mercari, Inc.\ and the RIISE.

\newpage
\bibliography{bibliography}

\begin{thebibliography}{10}

\bibitem{Azar2017-vj}
Yossi Azar, Ashish Chiplunkar, and Haim Kaplan.
\newblock Polylogarithmic bounds on the competitiveness of min-cost perfect matching with delays.
\newblock In {\em Proceedings of the 2017 Annual {ACM-SIAM} Symposium on Discrete Algorithms ({SODA} 2017)}, pages 1051--1061, 2017.

\bibitem{Azar2020-yp}
Yossi Azar, Ashish Chiplunkar, Shay Kutten, and Noam Touitou.
\newblock Set cover with delay - clairvoyance is not required.
\newblock In {\em Proceedings of the 28th Annual European Symposium on Algorithms ({{ESA}} 2020)}, pages 8:1--8:21, 2020.

\bibitem{Azar2020-in}
Yossi Azar and Amit Jacob~Fanani.
\newblock Deterministic min-cost matching with delays.
\newblock {\em Theory Comput. Syst.}, 64(4):572--592, 2020.

\bibitem{Azar2021-bb}
Yossi Azar, Runtian Ren, and Danny Vainstein.
\newblock The min-cost matching with concave delays problem.
\newblock In {\em Proceedings of the 2021 Annual {ACM-SIAM} Symposium on Discrete Algorithms ({SODA} 2021)}, pages 301--320, 2021.

\bibitem{Azar2019-sn}
Yossi Azar and Noam Touitou.
\newblock General framework for metric optimization problems with delay or with deadlines.
\newblock In {\em Proceedings of the 2019 {IEEE} 60th Annual Symposium on Foundations of Computer Science ({FOCS} 2019)}, pages 60--71, 2019.

\bibitem{Bienkowski2016-dk}
Marcin Bienkowski, Martin B{\"o}hm, Jaroslaw Byrka, Marek Chrobak, Christoph D{\"u}rr, Luk'avs Folwarczn'y, Lukasz Jez, Jiri Sgall, Kim~Thang Nguyen, and Pavel Vesel{\'y}.
\newblock Online algorithms for multi-level aggregation.
\newblock In {\em Proceedings of the 24th Annual European Symposium on Algorithms ({{ESA}} 2016)}, pages 12:1--12:17, 2016.

\bibitem{Bienkowski2014-fz}
Marcin Bienkowski, Jaroslaw Byrka, Marek Chrobak, Lukasz Jez, Dorian Nogneng, and Jir{\'\i} Sgall.
\newblock Better approximation bounds for the joint replenishment problem.
\newblock In {\em Proceedings of the 2014 Annual {ACM-SIAM} Symposium on Discrete Algorithms ({{SODA}} 2014)}, pages 42--54, 2014.

\bibitem{Bienkowski2018-yt}
Marcin Bienkowski, Artur Kraska, Hsiang-Hsuan Liu, and Pawel Schmidt.
\newblock A primal-dual online deterministic algorithm for matching with delays.
\newblock In {\em Proceedings of the 16th Workshop on Approximation and Online Algorithms ({WAOA} 2018)}, pages 51--68, 2018.

\bibitem{Buchbinder2017-qx}
Niv Buchbinder, Moran Feldman, Joseph~S. Naor, and Ohad Talmon.
\newblock {$O($depth$)$}-competitive algorithm for online multi-level aggregation.
\newblock In {\em Proceedings of the 2017 Annual {ACM-SIAM} Symposium on Discrete Algorithms ({SODA} 2017)}, pages 1235--1244, 2017.

\bibitem{Buchbinder2013-yc}
Niv Buchbinder, Tracy Kimbrel, Retsef Levi, Konstantin Makarychev, and Maxim Sviridenko.
\newblock Online make-to-order joint replenishment model: Primal-dual competitive algorithms.
\newblock {\em Oper. Res.}, 61(4):1014--1029, 2013.

\bibitem{Carrasco2018-vp}
Rodrigo~A. Carrasco, Kirk Pruhs, Cliff Stein, and Jos{\'e} Verschae.
\newblock The online set aggregation problem.
\newblock In {\em Proceedings of the 13th Latin American Symposium on Theoretical Informatics (LATIN 2018)}, pages 245--259, 2018.

\bibitem{Chen2022-qj}
Ryder Chen, Jahanvi Khatkar, and Seeun~William Umboh.
\newblock Online weighted cardinality joint replenishment problem with delay.
\newblock In {\em Proceedings of the 49th International Colloquium on Automata, Languages, and Programming ({{ICALP}} 2022)}, pages 40:1--40:18, 2022.

\bibitem{Deryckere2023-oj}
Lindsey Deryckere and Seeun~William Umboh.
\newblock Online matching with set and concave delays.
\newblock In {\em Proceedings of the Approximation, Randomization, and Combinatorial Optimization. Algorithms and Techniques ({{APPROX/RANDOM}} 2023)}, pages 17:1--17:17, 2023.

\bibitem{Dooly2001-yl}
Daniel~R. Dooly, Sally~A. Goldman, and Stephen~D. Scott.
\newblock On-line analysis of the {TCP} acknowledgment delay problem.
\newblock {\em J. ACM}, 48(2):243--273, 2001.

\bibitem{Echenique2023-pd}
Federico Echenique, Nicole Immorlica, and Vijay~V. Vazirani.
\newblock {\em Online and Matching-Based Market Design}.
\newblock Cambridge University Press, 2023.

\bibitem{Emek2016-sa}
Yuval Emek, Shay Kutten, and Roger Wattenhofer.
\newblock Online matching: Haste makes waste!
\newblock In {\em Proceedings of the {Forty-Eighth} Annual {ACM} Symposium on Theory of Computing ({STOC} 2016)}, pages 333--344, 2016.

\bibitem{Emek2019-if}
Yuval Emek, Yaacov Shapiro, and Yuyi Wang.
\newblock Minimum cost perfect matching with delays for two sources.
\newblock {\em Theor. Comput. Sci.}, 754:122--129, 2019.

\bibitem{He2023-ot}
Kun He, Sizhe Li, Enze Sun, Yuyi Wang, Roger Wattenhofer, and Weihao Zhu.
\newblock Randomized algorithm for {MPMD} on two sources.
\newblock In {\em Proceedings of the 19th Conference On Web And InterNet Economics ({WINE} 2023)}, pages 348--365, 2023.

\bibitem{Kakimura2023-im}
Naonori Kakimura and Tomohiro Nakayoshi.
\newblock Deterministic primal-dual algorithms for online k-way matching with delays.
\newblock In {\em Proceedings of the 29th International Computing and Combinatorics Conference ({COCOON} 2023)}, pages 238--249, 2023.

\bibitem{Karlin2001-fx}
Anna~R. Karlin, Claire Kenyon, and Dana Randall.
\newblock Dynamic {TCP} acknowledgement and other stories about $e/(e-1)$.
\newblock In {\em Proceedings of the 33rd Annual {ACM} Symposium on Theory of Computing ({STOC} 2001)}, pages 502--509, 2001.

\bibitem{Karp1990-kk}
Richard~M. Karp, Umesh~V. Vazirani, and Vijay~V. Vazirani.
\newblock An optimal algorithm for on-line bipartite matching.
\newblock In {\em Proceedings of the 22nd Annual {ACM} Symposium on Theory of Computing ({STOC} 1990)}, pages 352--358, 1990.

\bibitem{Liu2018-wb}
Xingwu Liu, Zhida Pan, Yuyi Wang, and Roger Wattenhofer.
\newblock Impatient online matching.
\newblock In {\em Proceedings of the 29th International Symposium on Algorithms and Computation ({ISAAC} 2018)}, pages 62:1--62:12, 2018.

\bibitem{Liu2022-bu}
Xingwu Liu, Zhida Pan, Yuyi Wang, and Roger Wattenhofer.
\newblock Online matching with convex delay costs.
\newblock {\em arXiv:2203.03335}, 2022.

\bibitem{Mehta2013-lp}
Aranyak Mehta.
\newblock Online matching and ad allocation.
\newblock {\em Found. Trends Theor. Comput. Sci.}, 8(4):265--368, 2013.

\bibitem{Mehta2007-cl}
Aranyak Mehta, Amin Saberi, Umesh~V. Vazirani, and Vijay~V. Vazirani.
\newblock {AdWords} and generalized online matching.
\newblock {\em J. ACM}, 54(5):22:1--22:19, 2007.

\bibitem{Melnyk2021-zv}
Darya Melnyk, Yuyi Wang, and Roger Wattenhofer.
\newblock Online {$k$-way} matching with delays and the {$H$-metric}.
\newblock {\em arXiv:2109.06640}, 2021.

\bibitem{Pavone2022-cg}
Marco Pavone, Amin Saberi, Maximilian Schiffer, and Matthew~W. Tsao.
\newblock Technical note---online hypergraph matching with delays.
\newblock {\em Oper. Res.}, 70(4):2194--2212, 2022.

\bibitem{Touitou2021-dx}
Noam Touitou.
\newblock Nearly-tight lower bounds for set cover and network design with deadlines/delay.
\newblock In {\em Proceedings of the 32nd International Symposium on Algorithms and Computation ({ISAAC} 2021)}, pages 53:1--53:16, 2021.

\end{thebibliography}

\begin{short}
\appendix 

\section{\texorpdfstring{$2$}{2}-competitive algorithm for \texorpdfstring{$f(n) = \lceil n / k \rceil$}{f(n)=ceil(n/k)}} \label{sec:appendix_2_comp}

\section{Lower Bound for match-all-remaining algorithms} \label{sec:simple_lower_bound}

\section{Omitted Proofs for  \texorpdfstring{\Cref{subsec:upper}}{Section 4.1}} \label{sec:proof_upper}

\lmmsivalwait*

\lmmoptcostwait*

\lmmalgcostphase*

\lmmoptcostphase*

\begin{proof}
    If an algorithm matches a number of requests that is not a multiple of $k$ within a phase, it immediately incurs a cost of $1$.
    Consequently, we only need to consider algorithms that do not match a set of requests whose size is not a multiple of $k$ within a phase. Furthermore, matching $k$ unmatched requests immediately does not incur a loss, so we consider such algorithms.

    An algorithm may carry over some requests from the previous phase. 
    For each $i \in \mathbb{Z}_k$, $W_i$ represents the waiting cost incurred by an algorithm that carries over $\overline{k - i}$ requests. This is because, when the number of requests already given in the phase is $s$, the number of requests left unmatched by an algorithm carrying over $\overline{k - i}$ requests is $\overline{s + k - i} = \overline{s - i}$, which equals the rate of increase of $W_i$ by the proposed algorithm.
    At the end of a phase, $W_i \geq 1$ for all $i \in \mathbb{Z}_k$. This implies that a waiting cost of at least $1$ is incurred regardless of the number of requests carried over.

    Therefore, any algorithm incurs a cost of at least $1$ within a phase.
\end{proof}

\section{Omitted Proofs for \texorpdfstring{\Cref{subsec:lower}}{Section 4.2}} \label{sec:proof_lower}

We show some lemmas that are used in the proofs of \Cref{lmm:adv_stepnumber}, \Cref{lmm:opt_upper_for_lb}, and \Cref{lmm:alg_lower_for_lb}.

\lmmadvstepnumber*

\lmmoptupperforlb*


    



\lmmalglowerforlb*

\begin{proof}
    The number of rounds in $\sigma_{\ALG}$ is at least $\alpha/2-1$ by \Cref{lmm:adv_stepnumber}.
    As the rounds move on to the next, when the algorithm matches a non-multiple of $k$ requests or $W_{\overline{a}}$ increases by $1$, the algorithm incurs at least a cost of $1$ in each round.
    Therefore, the total cost incurred over the instance is at least $\alpha/2-1$.
\end{proof}

\end{short}




\end{document}